\documentclass[11pt]{article}
\usepackage[utf8]{inputenc}
\usepackage[margin=1in]{geometry}
\usepackage[T1]{fontenc} 
\usepackage{graphicx}
\usepackage{mathpazo}
\usepackage{hyperref}
\hypersetup{
  colorlinks = true,  
  urlcolor = {blueGrotto},
  linkcolor = {royalBlue},
  citecolor = {navyBlue}
}

\usepackage{booktabs} 
\usepackage{multirow} 
\usepackage{multicol}

\usepackage[suppress]{color-edits} 
\addauthor[Anay]{am}{gold}
\addauthor[Jon]{jk}{limeGreen} 
\addauthor[Amin]{as}{purple} 
\addauthor[Grigoris]{gv}{blue} 

\usepackage{xcolor}
\definecolor{niceRed}{RGB}{190,38,38}
\definecolor{blueGrotto}{HTML}{059DC0}
\definecolor{royalBlue}{HTML}{057DCD}
\definecolor{navyBlue}{HTML}{0B579C}
\definecolor{limeGreen}{HTML}{81B622}
\definecolor{nicePurple}{HTML}{9c27b0}
\definecolor{lightRoyalBlue}{HTML}{def2ff} 
\definecolor{gold}{HTML}{ffa300}

\usepackage[utf8]{inputenc} 
\usepackage[T1]{fontenc}    
\usepackage{hyperref}       
\usepackage{url}            
\usepackage{booktabs}       
\usepackage{amsfonts}       
\usepackage{nicefrac}       
\usepackage{microtype}      
\usepackage{xcolor}         
 
\usepackage[
  backref=true,
  backend=biber,
  natbib=true,
  style=alphabetic,
  sorting=alphabeticlabel,
  sortcites=true,
  minbibnames=3,
  maxbibnames=999,
  mincitenames=1,
  maxcitenames=3,
  minalphanames=3,
  maxalphanames=4,
]{biblatex}

\DeclareSortingTemplate{alphabeticlabel}{
  \sort[final]{%
    \field{labelalpha} 
  }
  \sort{%
    \field{year}   
  }
  \sort{%
    \field{title}  
  }
}

\AtBeginRefsection{\GenRefcontextData{sorting=ynt}}
\AtEveryCite{\localrefcontext[sorting=ynt]}
\usepackage{microtype} 
\usepackage{subfigure} 
\usepackage{booktabs} %
 
\usepackage{hyperref}
 
\usepackage{algpseudocode} 
\usepackage[linesnumbered,ruled,vlined]{algorithm2e}
\SetKwInOut{Oracles}{Query Access}

\usepackage{amsmath}
\usepackage{amssymb}
\usepackage{mathtools}
\usepackage{amsthm}
\usepackage{bbm} 
\usepackage[capitalise,noabbrev,nameinlink,sort]{cleveref} 
\usepackage{nicefrac}
\usepackage{thm-restate}

\usepackage[textsize=small]{todonotes}
\usepackage{mathrsfs}
\usepackage{mathtools} 
\usepackage{dsfont}
\usepackage[framemethod=TikZ]{mdframed}
\mdfsetup{%
backgroundcolor=white, 
linewidth=0.5,
skipabove=1pt,
skipbelow=-2pt,
innertopmargin=2pt,
innerbottommargin=2pt,
roundcorner=4pt}

\usepackage{pifont} 
\usepackage{color,multicol} 
\usepackage[inline]{enumitem} 
\usepackage{multirow}

\usepackage{centernot}

\usepackage[many]{tcolorbox} 

\usepackage[normalem]{ulem}
\usepackage{comment}

\usepackage[T1]{fontenc} %
 
\usepackage{xcolor}                 %
\definecolor{bubblegreen}{RGB}{190,38,38}
\definecolor{bubblegray}{RGB}{241,240,240}
 
\usepackage{mathrsfs,dsfont,bbm}    %
\usepackage{multicol,multirow}      %
\usepackage[inline]{enumitem}       %
\usepackage[many]{tcolorbox}        %
\usepackage[normalem]{ulem}         %

\theoremstyle{plain} 
\newtheorem{theorem}{Theorem}[section]

\newtheorem{proposition}[theorem]{Proposition}
\newtheorem{lemma}[theorem]{Lemma}
\newtheorem{claim}[theorem]{Claim}
\newtheorem{fact}[theorem]{Fact}

\newtheorem{inftheorem}[theorem]{Informal Theorem}

\newtheorem{definition}{Definition}
\newtheorem{infdefinition}{Informal Definition}
\newtheorem*{definition*}{Definition}

\theoremstyle{definition} 
\newtheorem{example}[theorem]{Example}
\newtheorem{remark}[theorem]{Remark}

\theoremstyle{remark}
\newtheorem{observation}[theorem]{Observation}

\crefname{section}{Section}{Sections}
\crefname{theorem}{Theorem}{Theorems}
\crefname{assumption}{Assumption}{Assumptions}
\crefname{lemma}{Lemma}{Lemmas}
\crefname{definition}{Definition}{Definitions}
\crefname{conjecture}{Conjecture}{Conjectures}
\crefname{corollary}{Corollary}{Corollaries}
\crefname{construction}{Construction}{Constructions}
\crefname{conjecture}{Conjecture}{Conjectures}
\crefname{claim}{Claim}{Claims}
\crefname{observation}{Observation}{Observations}
\crefname{proposition}{Proposition}{Propositions}
\crefname{fact}{Fact}{Facts}
\crefname{question}{Question}{Questions}
\crefname{problem}{Problem}{Problems}
\crefname{remark}{Remark}{Remarks}
\crefname{example}{Example}{Examples}
\crefname{equation}{Equation}{Equations}
\crefname{appendix}{Appendix}{Appendices}
\crefname{algorithm}{Algorithm}{Algorithms}
\crefname{model}{Model}{Models}
\crefname{figure}{Figure}{Figures}

\AfterEndEnvironment{definition}{\noindent\ignorespaces}
\AfterEndEnvironment{infdefinition}{\noindent\ignorespaces}
\AfterEndEnvironment{example}{\noindent\ignorespaces}
\AfterEndEnvironment{assumption}{\noindent\ignorespaces}
\AfterEndEnvironment{lemma}{\noindent\ignorespaces}
\AfterEndEnvironment{theorem}{\noindent\ignorespaces}
\AfterEndEnvironment{proposition}{\noindent\ignorespaces}
\AfterEndEnvironment{fact}{\noindent\ignorespaces}
\AfterEndEnvironment{question}{\noindent\ignorespaces}
\AfterEndEnvironment{corollary}{\noindent\ignorespaces}
\AfterEndEnvironment{model}{\noindent\ignorespaces}
\AfterEndEnvironment{remark}{\noindent\ignorespaces}
\AfterEndEnvironment{proof}{\noindent\ignorespaces}
\AfterEndEnvironment{fact}{\noindent\ignorespaces}
\AfterEndEnvironment{minftheorem}{\noindent\ignorespaces}
\AfterEndEnvironment{inftheorem}{\noindent\ignorespaces}
\AfterEndEnvironment{maintheorem}{\noindent\ignorespaces}
\AfterEndEnvironment{restatable}{\noindent\ignorespaces}
\AfterEndEnvironment{observation}{\noindent\ignorespaces}

\newcommand{\eat}[1]{}

\makeatletter
\newcommand{\customlabel}[2]{%
\protected@write \@auxout {}{\string \newlabel {#1}{{#2}{\thepage}{#2}{#1}{}} }%
\hypertarget{#1}{}
}
\makeatother

\newcommand{\quadtext}[1]{\quad\text{#1}\quad}
\newcommand{\qquadtext}[1]{\qquad\text{#1}\qquad}
 
\newcommand{\quadand}{\quadtext{and}}
\newcommand{\qquadand}{\qquadtext{and}}

\newcommand{\qquadwhere}{\qquadtext{where}}

\def\abs#1{\left| #1 \right|}

\newcommand{\sinbrace}[1]{\{#1\}}

\newcommand{\inbrace}[1]{\left\{#1\right\}}

\newcommand{\inparen}[1]{\left(#1\right)}
\newcommand{\insquare}[1]{\left[#1\right]}

\newcommand{\floor}[1]{\left\lfloor#1\right\rfloor}

\newcommand{\N}{\mathbb{N}}

\newcommand{\nfrac}[2]{\nicefrac{#1}{#2}}
\newcommand{\sfrac}[2]{{#1/#2}}

\renewcommand{\epsilon}{\varepsilon} 

\makeatletter
\def\mov@rlay#1#2{\leavevmode\vtop{%
  \baselineskip\z@skip
  \lineskiplimit-\maxdimen
  \ialign{\hfil$\m@th#1##$\hfil\cr#2\crcr}}}

\newcommand{\charfusion}[3][\mathord]{%
  #1{%
    \ifx#1\mathop\vphantom{#2}\fi
    \mathpalette\mov@rlay{#2\cr#3}%
  }%
  \ifx#1\mathop\expandafter\displaylimits\fi
}
\makeatother

\makeatletter
\newcommand*{\tran}{{\mathpalette\@tran{}}}
\newcommand*{\@tran}[2]{\raisebox{\depth}{$\m@th#1\intercal$}}
\makeatother

\renewcommand{\bar}{\overline}

\def\<{\langle}
\def\>{\rangle}

\DeclareMathAlphabet{\mathpzc}{OT1}{pzc}{m}{it}

\newcommand{\customcal}[1]{\euscr{#1}}
\newcommand{\cA}{\customcal{A}}

\newcommand{\cC}{\customcal{C}}

\newcommand{\cF}{\customcal{F}}

\newcommand{\cI}{\customcal{I}}

\newcommand{\cL}{\customcal{L}}
\newcommand{\cM}{\customcal{M}}

\newcommand{\cP}{\customcal{P}}

\newcommand{\cX}{\customcal{X}}

\DeclareMathAlphabet{\mathdutchcal}{U}{dutchcal}{m}{n}
\SetMathAlphabet{\mathdutchcal}{bold}{U}{dutchcal}{b}{n}
\DeclareMathAlphabet{\mathdutchbcal}{U}{dutchcal}{b}{n}
 
\DeclareMathAlphabet\urwscr{U}{urwchancal}{b}{n}%
\DeclareMathAlphabet\rsfscr{U}{rsfso}{m}{n}
\DeclareMathAlphabet\euscr{U}{eus}{m}{n}
\DeclareFontEncoding{LS2}{}{}
\DeclareFontSubstitution{LS2}{stix}{m}{n}
\DeclareMathAlphabet\stixcal{LS2}{stixcal}{m} {n}

\renewcommand{\paragraph}[1]{\bigskip  \noindent\textbf{#1}~~}

\newcommand{\ie}{\textit{i.e.}}

\usepackage{tikz}
\usetikzlibrary{patterns}

\newcommand{\algo}[1]{\mathpzc{#1}}
\newcommand{\generator}{\algo{G}}

\usepackage{array}
\newcolumntype{L}[1]{>{\raggedright\let\newline\\\arraybackslash\hspace{0pt}}m{#1}}
\newcolumntype{C}[1]{>{\centering\let\newline\\\arraybackslash\hspace{0pt}}m{#1}}
\newcolumntype{R}[1]{>{\raggedleft\let\newline\\\arraybackslash\hspace{0pt}}m{#1}}

\usepackage{soul}

\newmdenv[
    backgroundcolor=lightgray!10, %
    roundcorner=5pt,            %
    linecolor=black,             %
    linewidth=1pt,               %
    innertopmargin=5pt,         %
    innerbottommargin=0pt,      %
    innerleftmargin=10pt,        %
    innerrightmargin=10pt,       %
    skipabove=5pt,              %
    skipbelow=0pt               %
]{curvybox}

\makeatletter
\@ifundefined{Require}{}{}
\@ifundefined{Ensure}{}{}
\@ifundefined{State}{}{}
\@ifundefined{Return}{}{}
\@ifundefined{For}{}{}
\@ifundefined{EndFor}{}{}
\@ifundefined{If}{}{}
\@ifundefined{Else}{}{}
\@ifundefined{EndIf}{}{}
\@ifundefined{While}{}{}
\@ifundefined{EndWhile}{}{}
\@ifundefined{Comment}{}{}
\makeatother

\newcommand{\learner}{\algo{A}}

\title{On Language Generation in the Limit with Bounded Memory}

\date{}
\author{
        \begin{tabular}{C{7.5cm}C{7.5cm}}
        {\bf Jon Kleinberg}
            & {\bf Anay Mehrotra}\\
        {Cornell University} 
            & {Stanford University}\\
        \mbox{{\href{mailto:steve.hanneke@gmail.com }{kleinberg@cornell.edu}}} 
            &   \mbox{{\href{mailto:anaymehrotra1@gmail.com}{anaymehrotra1@gmail.com}}}\\[4mm]
        {\bf Amin Saberi}
            & {\bf Grigoris Velegkas}\\
                {Stanford University} 
                    & {Google Research}\\
        \mbox{{\href{mailto:amin.saberi@gmail.com}{amin.saberi@gmail.com}}}
            & 
                \mbox{{\href{mailto:gvelegkas@google.com}{gvelegkas@google.com}}} 
        \end{tabular}
}

\begin{document}
\maketitle
 
\thispagestyle{empty}

\begin{abstract}
    We study language generation in the limit under bounded memory. 
    In this task, a learner observes examples from an unknown target language one at a time and must eventually output only new valid examples from the target.

    {Existing work assumes the learner has access to the entire example history which is a strong assumption, since algorithms operating under realistic resource constraints can retain only limited information about the past.} {A large body of work in learning theory shows that memory constraints can dramatically change what is learnable; we extend this line of work to language generation in the limit.}

    First, we first study generators that remember no past examples.
    {We show that under a mild restriction on the enumeration, every} countable collection of infinite languages remains generable even with no memory.
    {Without this restriction, we give an exact characterization of when memoryless generation is possible.} {We then turn to quantitative coverage guarantees for finite collections, characterizing the optimal minimax density achievable by memoryless generators --- the best density one can guarantee against any collection of a given size.} This minimax bound is combinatorial and {relies on Sperner's theorem together with symmetric chain decompositions}.
    We further show that giving the learner a sliding window of the last $W$ examples does not improve this worst-case density, whereas allowing it to store $b$ adaptively chosen past examples strictly improves the achievable density for every $b \geq 1$.
    
    Finally, we revisit identification in the limit, the classical task where the learner must converge to a single correct hypothesis for the target language. 
    We focus on its well-studied incremental variant, where the learner remembers only its previous guess. 
    Here, although exact identification already fails on a collection of just three languages, we show that a mild relaxation, which only requires convergence to an ``approximate'' version of the target, is achievable for every finite collection.
    {Together, these results show that bounded memory affects the three tasks differently: generation remains achievable for every countable collection, while density and identification are confined to finite collections, with guarantees that weaken as the collection grows.}
 
\end{abstract}

\newpage
\thispagestyle{empty}
\tableofcontents
\newpage 

\pagenumbering{arabic} 

\vspace{-3mm}
\section{Introduction}
\label{sec:intro}
\vspace{-2.5mm}
%
Modeling how languages are learned has long fascinated theoretical computer scientists, and one of the central formal models in this direction is \emph{language identification in the limit} \citep{gold1967language,angluin1979finding,angluin1980inductive}.
In this task, given a positive enumeration of an unknown target language, the learner repeatedly outputs a hypothesis for the language itself and succeeds if these hypotheses eventually stabilize on the correct target language. 
{This requirement turns out to be quite stringent: Angluin's \citep{angluin1980inductive} results show that identification is achievable only for very limited language families, ruling out even simple infinite collections such as the regular languages.}

This model of language identification was itself motivated in part by the question of how humans acquire language from positive examples. Building on this perspective, \citet{kleinberg2024language} {showed that} if one weakens the learner's goal from identifying the target language to \textit{generating} fresh correct examples from it (which is closer to the real-world requirements of the task), the problem becomes {significantly} {more tractable.}

They termed this weaker requirement 
\emph{language generation in the limit.}
More concretely, fix a countable collection of infinite languages $\cL=\inbrace{L_1,L_2,\ldots}$ over a countable domain $\cX$. An adversary chooses a target language $K\in\cL$ and reveals an enumeration $x_1,x_2,\ldots$ of $K$; after round $t$, the learner has seen the sample $S_t=\inbrace{x_1,\ldots,x_t}$.
It then outputs a new string $y_t$ and its goal is that after some finite time $t^\star$ every output $y_t$ lies in $K\setminus S_t$. \citet{kleinberg2024language} showed that this weaker task is possible for \emph{every} countable collection of languages.
The contrast between these two results is {stark}. 
In both settings, the learner receives the same information: a stream of positive
examples from the target language.  
The only difference is what it must produce. 
Yet, this difference has dramatic consequences:
Identification is possible only for highly
restricted families of {countable} languages, while generation is possible for
all of them.

{This stark contrast raises a natural question: how much more can we demand of a generator before its task becomes as hard as identification? A line of subsequent work has explored this by introducing intermediate notions in which the learner must cover ``dense'' subsets of the target language \cite{kleinberg2025density,kalavasis2026characterizations,charikar2025facets,kleinberg2025partial}. These notions form a hierarchy of tasks between generation and identification, graded by how much of the target language the learner must cover.}

{All of these results}, however, share a common assumption: the learner has access to the entire interaction {history---at each round, its output depends on the full sample $S_t$, and existing constructions rely heavily on this.
This is a strong assumption on two fronts.} {Practically, it does not reflect how humans or modern machine learning systems generate language, since neither retains its full ``training data.''} {Theoretically, it obscures a basic question: how much information about the past is actually necessary for successful generation? Both concerns naturally motivate studying language generation under bounded memory.}

\enlargethispage{\baselineskip}

{Bounded-memory restrictions have a long history in learning theory and theoretical computer science.
In learning theory, for instance, they appear in active learning \cite{hopkins2021bounded}, in estimation \cite{steinhardt2016memory,raz2016fast,raz2017time,sharan2019memory}, and in online learning \cite{srinivas2022memory,peng2023near}.
Most directly relevant to our setting, bounded memory has even been studied for language identification in the limit, through certain iterative and incremental identification models which we discuss in more detail later \cite{lange1996incremental,case1999incremental}.}
A recurring lesson from these works is that bounds on memory can qualitatively change what is learnable.
This naturally raises the following question for language generation:
\vspace{-1mm}
\begin{center}
    \textit{Which collections can be generated in the limit and what density\\[-0.5mm] guarantees are achievable when the learner has bounded memory?}
\end{center}

\subsection{Model and Results}
In this section, we state informal statements of our main results.
\subsubsection{Generation, Density, and Identification}
\label{sec:intro:notions}
{As discussed above, language generation in the limit (or simply \textit{generation}) is the easiest task in this hierarchy, and language identification in the limit (or simply \textit{identification}) is the hardest.}

{Formally, let} $\cL=\inbrace{L_1,L_2,\ldots}$ be a countable collection of infinite languages over a countable domain $\cX$, where every $L_i \subseteq \cX$. An adversary chooses a target language $K=L_z\in\cL$ and reveals an \emph{enumeration} $x_1,x_2,\ldots$ of $K$, meaning that every element of $K$ appears at least once, though repetitions are allowed. After round $t$, the learner has seen the sample $S_t=\inbrace{x_1,\ldots,x_t}$.
Then, generation is defined as follows.
\begin{infdefinition}[Generators; see \cref{sec:model-prelim}]
    A generator can be of two types: {it is \ul{element-based} if its output at time $t$ is an element $y_t\in \cX$, and it is \ul{set-based} if its output at time $t$ is an infinite set $G_t\subseteq \cX$.}
\end{infdefinition}
\vspace{-7mm}
\begin{definition}[Generation in the limit]
    {The generator $\generator$ is said to \emph{generate} $K$ in the limit if, for every enumeration of $K$, there is a time $t^\ast$ such that every output after time $t^\ast$ is valid for $K$. In the element-based case, validity means $y_t\in K\setminus S_t$; in the set-based case, it means that the output is an infinite set $G_t\subseteq K$. The generator $\generator$ is said to generate a collection $\cL = \inbrace{L_1,L_2,\ldots}$ if it generates every $K \in \cL$.}
\end{definition}
{In other words, after some finite time, an element-based generator must always produce an unseen element of $K$, while a set-based generator must always produce an infinite subset of $K$.}
{As mentioned above, \citet{kleinberg2024language} constructed a generator that can generate} from any countable collection of languages $\cL$ in the limit.
This generator, however, suffers from ``mode collapse:'' it {may keep generating from smaller and smaller subsets of $K$, covering only a vanishing fraction of the language}.
{In other words, generation in the limit guarantees validity, but not coverage of $K$.}

This motivated several works to study the next task in the hierarchy: \emph{generation with ``breadth'' or density}.
{Dense generation strengthens generation by requiring the learner's outputs to cover a non-trivial fraction of the target language.}
To speak about density, it is convenient to think of the learner at round $t$ as outputting a set $G_t$.
One may then ask not just whether $G_t\subseteq K$, but whether it captures a non-trivial fraction of the target language.
Several recent works have proposed quantitative notions of density for measuring how much of the target language the generator covers \cite{kalavasis2025limits,charikar2025facets,kalavasis2026characterizations,kleinberg2025density}.
{Many of these notions are strong enough to become equivalent to identification or close variants of it. By contrast, the density notions introduced by \citet{kleinberg2025density} and subsequently studied by \citet{kleinberg2025partial,mehrotra2025language} are more tractable.}
In our work, we follow this density viewpoint.
We define and discuss density more carefully in \cref{sec:intro:dense-generation}.
For now, consider the following example: if $K=\mathbb{N}$ and at some round the learner outputs $G_t=2\mathbb{N}=\inbrace{2,4,\ldots}$, then $G_t$ has density $\mu\inparen{G_t;K}=1/2$ inside $K$.
We say that a generator achieves density $1/2$ if it outputs sets $G_t$ with $\mu\inparen{G_t;K}\geq 1/2$ for infinitely many values of $t$.

{Finally, we turn to the hardest task in the hierarchy, \emph{identification}.} 

\begin{definition}[Identification in the limit; \cite{gold1967language}]
    {The identifier $\cI$ is said to \emph{identify} $K$ in the limit if, for every enumeration of $K$, there is a time $t^\ast$ such that for every $t\geq t^\ast$ its output $i_t$ satisfies $L_{i_t}=K$. The identifier $\cI$ is said to identify a collection $\cL = \inbrace{L_1,L_2,\ldots}$ if it identifies every $K \in \cL$.}
\end{definition}
{In other words, identification asks the learner to recover the target language entirely. Generation only requires the learner to produce new elements of $K$, and $\rho$-dense generation asks it to produce sets covering a $\rho$-fraction of the target.}

{Identification is much more restrictive: \citet{gold1967language} showed that every finite collection is identifiable, while certain countable collections are not.}
Later, work by \citet{angluin1980inductive} gives a complete characterization of identifiable collections; this characterization was largely negative and showed that identifiability does not extend much beyond finite collections.
Hence, in some sense, finite collections are the most natural family of languages for which identification is possible, and they serve as a testbed for identification results when one imposes additional constraints.

{With this hierarchy in place, we ask how it changes under bounded memory.}

\subsubsection{Two Types of Memory and the Resulting Models}
\label{sec:intro:bounded-memory}

{To make this question precise, we must first specify what information the learner can carry from one round to the next. In the unbounded-memory model, the learner conditions its output on the entire sample history $S_t$. Once memory is bounded, only some of this information can survive.}
Here, we separate two conceptually different types of information.
\begin{enumerate}
    \item Information from the environment: The learner may remember previously seen examples from the adversary's enumeration.
    \item Information from the learner's past actions:
        A learner may forget the raw examples it has seen and still retain some internal state through its previous outputs.
\end{enumerate}
{With unbounded memory, the first kind of information is strictly more powerful than the second: a learner that remembers all observed examples can reconstruct its past actions, while the same action sequence could have arisen from many different example histories.}
{With limited memory, however, both kinds of information matter. Past examples are direct evidence about the target language. Past outputs, by contrast, were computed from the entire history seen at the time, so storing a past output is in effect storing a summary of that history.}

{To isolate the contribution of each resource, we study them separately. First, we consider a fully memoryless learner that only ``reacts'' to the current example (\cref{sec:intro:memoryless}). Next, we add two natural forms of limited memory over past examples: a sliding window of the last $W$ examples and a buffer of $b$ examples chosen by the learner. The contrast between them (keeping recent examples versus keeping chosen ones) isolates whether selecting which examples to remember is meaningful (\cref{sec:intro:dense-generation}). 
Finally, we strip away memory of past examples and allow only the simplest form of memory over past outputs: the learner's most recent guess (\cref{sec:intro:identification}).}
This last model has also been extensively studied in the iterative and incremental identification literature \cite{lange1996incremental,case1999incremental}.

\subsubsection{Generation with No Memory}
\label{sec:intro:memoryless}

{We begin with the most extreme setting outlined above: a generator that carries no state from one round to the next. It remembers neither past examples nor its own past outputs, and the current example is its only source of information.}

{\begin{restatable}[Memoryless generator]{definition}{memorylessgeneratordef}
\label{def:memoryless-generator}
Fix an output space $\Omega$. A generator with output space $\Omega$ is \emph{memoryless} if it is specified by a deterministic function $\generator:\cX\to\Omega$. On round $t$, after observing the current example $x_t$, the generator outputs $\generator(x_t)$.
\end{restatable}}
In particular, a memoryless generator is not aware of the round number and has no internal state.
{Depending on the output model, $\Omega$ is either $\cX$ (for element-based generation) or $\cX^\infty$ (for set-based generation, where $\cX^\infty$ denotes the family of all infinite subsets of $\cX$).}

{Since a memoryless generator remembers neither past examples nor its own past outputs, generating from any interesting collection seems hopeless.}
Indeed, suppose that there is a single ``bad'' element $b$ {on which the generator's output is invalid for $K$}.
Then, the adversary can construct a hard enumeration for $\generator$ by \emph{interleaving} copies of $b$ with an otherwise complete enumeration of $K$, producing a stream of the form $x_1,b,x_2,b,x_3,b,\ldots$ and {thereby forcing} the same mistake infinitely many times.

{Thus, without restricting repetitions, the adversary can turn a single mistake into infinitely many. This motivates the following mild restriction on the adversary in the memoryless model.}
\begin{definition}[Finitely repeating enumerations]
An enumeration $(x_n)_{n\in\N}$ of an infinite language $K\subseteq \cX$ is \emph{finitely repeating} if every element of $\cX$ appears only finitely many times in the sequence.
\end{definition}
Surprisingly, under finitely repeating enumerations, we obtain the following positive result:
\begin{restatable}{theorem}{memorylessfiniterepsthm}
\label{thm:memoryless-finite-reps}
Every countable collection of infinite languages admits a memoryless set-based generator under finitely repeating enumerations.
\end{restatable}
{This extends the positive result of \citet{kleinberg2024language} from full-memory to fully memoryless generators. In particular, it shows that the tractability of language generation does not stem from having access to the entire interaction history.}

{The set-based output is essential here (\cref{thm:memoryless-set-based-necessary}). If, given $x$, the generator could only output a single new string $y$, it would have no way of telling from the current example $x$ alone whether that string had appeared earlier (and this allows the adversary to always say $y$ before $x$ and force the generator to make a mistake). Outputting an infinite subset of the target avoids this: after viewing finitely many examples, at least one element of the set could still be unseen.}

{One may also ask what happens when the adversary is allowed to repeat elements infinitely often. We give a tight characterization: memoryless generation is possible if and only if, in the full-memory set-based model, every language in $\cL$ can already be generated after seeing any \textit{single} example (\cref{thm:memoryless-with-reps}).}

We overview the proof of \cref{thm:memoryless-finite-reps} in \cref{sec:overview}; formal details of this setting are in \cref{sec:memoryless:set}.

\subsubsection{{Density with No Memory, Sliding-Window Examples, and Chosen Examples}}
\label{sec:intro:dense-generation}
{The previous section established that memoryless generation is possible in the limit (\cref{thm:memoryless-finite-reps}). In other words, that there is an algorithm which in the limit produces new unseen examples from the target language and \textit{never} outputs examples from outside it.
    However, the outputs of this generator may still cover only a vanishing fraction of the target.
    We now ask what density guarantees are possible under bounded memory.

We study three memory models: the memoryless model from the previous section, a sliding window of the last $W$ examples, and a buffer of $b$ examples chosen by the learner. For each, we ask for the best density a generator can guarantee against any collection of a given size. It turns out that no positive density is achievable for arbitrary countable collections under bounded memory, so we focus on finite collections. As we noted earlier, finite collections serve as the natural testbed for identification under additional constraints; here we use them analogously for density.}
	    {The answer is governed by the ``width'' of the Boolean lattice: For $z\in \N$, define}
	    \[
	        \mu(z)\coloneqq \binom{z}{\floor{z/2}}\,.
	    \]
	    {This is the size of the largest antichain in the lattice $2^{[z]}$. (Recall that an antichain in $2^{[z]}$ is a collection of subsets of $\inbrace{1,\dots,z}$ in which no subset contains another.)}
	     \begin{inftheorem}[Minimax Density Guarantees for Finite Collections with Limited Memory]\label{thm:intro-density}
	            The following hold.
	            \begin{enumerate}
	                \item \textbf{No memory; see \cref{thm:minimax-upper-density-memoryless-finite}.} Fix $k \in \N$. For every finite collection $\cL$ of size $k$ there exists a generator that generates from $\cL$ in the limit and for infinitely many time steps the set it outputs has density at least {$\nfrac{1}{\mu(k-1)}$} in the target language. Conversely, there is a finite collection of size $k$ for which this bound cannot be improved.
	                \item \textbf{Sliding window of $W$ examples; see \cref{thm:window-upper-density-exact}.} Fix $k, W \in \N$. For every finite collection $\cL$ of size $k$ there exists a generator that generates from $\cL$ in the limit and for infinitely many time steps the set it outputs has density at least {$\nfrac{1}{\mu(k-1)}$} in the target language. Conversely, there is a finite collection of size $k$ for which this bound cannot be improved.
	                \item \textbf{$b$ chosen examples; see \cref{thm:no-eviction-universal-lower-bound}.} Fix $k, b \in \N$. For every finite collection $\cL$ of size $k$ there exists a generator that generates from $\cL$ in the limit and for infinitely many time steps the set it outputs has density at least 
	                \[
	                    \begin{cases}
\displaystyle {\frac{1}{\mu(k-b-1)}}\,, & 0\leq b\leq k-3\,,\\[3mm]
	1\,, & b\geq k-2\,.
\end{cases}
	                \]
	            \end{enumerate}
	        \end{inftheorem}
        {The three results together give {an interesting} picture{:} The first establishes the baseline: memoryless generators can guarantee density $1/\mu(k-1)$ and this is the optimal in a minimax sense. The second shows that a sliding window does not help: the worst-case minimax density remains $1/\mu(k-1)$, even with access to the last $W$ examples for any finite $W$. The third shows that adaptively chosen examples do help: each stored example effectively removes one language from the collection, so the memoryless bound applies to a residual collection of size $k-b$ instead of $k$.

The contrast between the sliding window and the adaptive buffer has a simple explanation. With a sliding window, the adversary can insert long blocks of uninformative examples between informative ones, so that the window never holds more than one informative example at a time. With an adaptive buffer, this is not an obstacle as the learner can keep each informative example once it appears, regardless of what ~ later.}
 
	    We remark that the result in the {sliding-window model} holds even if the enumeration does not contain any repetitions. 
	    {We give an overview of the proof of \cref{thm:intro-density} in \cref{sec:overview}; formal details of this setting are in \cref{sec:density-bounds}.}

\subsubsection{Approximate Identification with Last-Guess Memory}
\label{sec:intro:identification}

{Finally, we turn to identification, the strongest task in the hierarchy. As outlined above, we focus on the simplest memory setting: the learner forgets all observed examples and remembers only its most recent guess.}

	{Starting from \citet{lange1996incremental}, a long line of work has studied this problem, with largely negative results. As we show in \cref{prop:incremental-not-identifiable}, exact identification can fail even for collections of just three languages. Motivated by this obstruction, we consider a mild relaxation: \emph{approximate identification}, where the learner may converge to a hypothesis that differs from $K$ on at most finitely many elements. While close variants of this notion have been considered in the literature, to our best knowledge approximate identification itself has not been studied. Surprisingly, we show that approximate identification is possible for all finite collections:}
    \begin{theorem}[Incremental Approximate Identification]\label{thm:intro-incremental}
            All finite collections are approximately identifiable by a learner that is only allowed to remember its most recent output.
        \end{theorem}
        \noindent The formal details of this result appear in \cref{sec:incremental-bounded-memory}.

\subsection{Proof Overviews}
\label{sec:overview}
In this section we explain the main ideas behind our proofs. 

\paragraph{Proof overview of \cref{thm:memoryless-finite-reps}.} 
Perhaps the most natural idea is to try to ``reduce'' the 
memory consumption of algorithms that work
in the full-memory setting, like the one from \citet{kleinberg2024language}. In a nutshell, at every time step $t$, this
algorithm keeps track of the ``version space,'' \ie, the set
of all languages consistent with $S_t$, and extracts a descending chain of languages from the version space, which they term ``critical'' languages. Their main insight is that, {for large enough $t$, if one descends sufficiently far down this chain, the algorithm can output} some language $L_{i_t} \subseteq K.$

Unfortunately, in the memoryless setting we consider, it is prohibitive to estimate the version space, since it requires knowledge of \emph{every} past example. In fact, it is not hard to see that this is not merely an artifact of the completely memoryless setting; other reasonable settings of bounded memory also {face the same obstacle}, since approximating the version space requires continuously increasing memory. Another complication is that in order to figure out how far down the chain one needs to descend, it is important to use information about the current round of the game. For instance, a simple rule is to say that in round $t$ one descends down $t$ levels of the chain. Notice that learners in our setting do not have this information, as they are not even aware of the round of the game! 

Our approach to circumvent these issues is the following: when the learner sees an example $x$ it computes the quantity
\[
J_n(x)\coloneqq \bigcap \inbrace{L_j : j\leq n \text{ and } x\in L_j} {\quad\text{for }n\leq x},
\]
\ie, the intersection of appropriately long prefixes that are consistent with the singleton element $x$.\footnote{Here, we assume that $x$ can be mapped to a natural number so the comparison with the index $n$ is meaningful.}
Then, it outputs $J_{n(x)}(x)$ where $n(x) \leq x$ is the largest
number for which $\abs{J_{n(x)}(x)} = \infty.$ Our main insight is that for every $L_z \in \cL$ there are only finitely many {$x$} for which $n(x) < z.$ To see that, notice that there are finitely many languages that appear before $L_z.$ Let $U_z$ be the union of all intersections $\cap \cF$ such that $\cF \subseteq \inbrace{L_1,\ldots,L_z}$ and $\cap \cF$ is finite; since there are only finitely many such $\cF$, it follows that $U_z$ is finite. Moreover, we can
show that if $n(x) < z$ for infinitely many $x$, then $U_z$ needs to be infinite, leading to a contradiction. The formal
details appear in \cref{sec:memoryless:set}.

\paragraph{Proof overview of \cref{thm:intro-density}.} These results are the most technically involved components of our work. Our proofs all follow the same high-level principle: with limited memory, the generator can only output sets that are simultaneously ``safe'' for the languages still consistent with the information currently visible, so the problem reduces to understanding how large those safe regions can be. In the memoryless model this leads to an exact Sperner-type minimax value; in the sliding-window model a modified construction shows that a finite window does not enlarge the worst-case safe region; {finally, in the model where the generator can store $b$ examples of its choice}, each stored example permanently shrinks the ``ambiguity,'' and by choosing which examples to store the generator can, essentially, reduce the level of uncertainty to one corresponding to a collection of smaller size. We now explain each of the three parts of the result in more detail.

\paragraph{Memoryless upper density.} 
For the {upper-bound construction} in \cref{thm:upper-density-collapse}, we build a worst-case family using the ``middle'' layer of the Boolean lattice. Set
\[
n\coloneqq k-1
\qquadand 
N\coloneqq \binom{n}{\floor{n/2}}\,.
\]
We partition the target language $K$ into $N$ pairwise disjoint infinite sets
$
K=A_1\cup\cdots\cup A_N
$
with
$
\mu_{\rm up}(A_i;K)=\sfrac{1}{N}$ for all $i \in [N].$
We then index these pieces by the $N$ subsets
$
S_1,\dots,S_N\subseteq [n]
$
of cardinality exactly $\floor{n/2}$, so that $\inbrace{S_1,\dots,S_N}$ is an antichain by \cref{fact:sperner}. For each $j\in [n]$, we define
\[
L_j \coloneqq \bigcup_{i:\, j\in S_i} A_i\,.
\]
Now fix any memoryless generator that succeeds on the resulting collection. Outside finitely many bad points, if the current example lies in $A_i$, then the output must be contained in every $L_j$ with $j\in S_i$. By the antichain property, we can show that the corresponding intersection reduces back to $A_i$. Hence, every sufficiently late output has upper density at most $1/N$ inside $K$, which gives the desired upper bound.

The matching lower bound in \cref{thm:sperner-achievability} is achieved by the canonical memoryless intersection generator: on input $x$, it outputs the intersection of all languages containing $x$ whenever that intersection is infinite. Fix a target language $K\in \cL$. For each subset
$
B\subseteq \cL\setminus \inbrace{K},
$
we consider the relative region
\[
R'_B\coloneqq \inbrace{x\in K : S_K(x)=B}\,.
\]
These regions partition $K$. We then apply the symmetric chain decomposition from \cref{fact:symmetric-chain} to the Boolean lattice on $\cL\setminus\inbrace{K}$, obtaining exactly
$N$
chains. Since the corresponding chain-unions partition $K$, finite subadditivity of $\mu_{\rm up}$ implies that some chain-union has upper density at least $1/N$. Along that chain, once we pass the finitely many finite regions, the remaining union is contained in a single intersection $I'_B$ that, therefore, also has upper density at least $1/N$. Because the corresponding region $R'_B$ is infinite, every {finitely repeating enumeration of $K$} visits it infinitely often, and at each such time the generator outputs exactly $I'_B$. This proves the matching lower bound and yields \cref{thm:minimax-upper-density-memoryless-finite}.

\paragraph{Sliding-window.}  The upper bound in \cref{lem:window-upper-density-collapse} is obtained by modifying the previous Sperner construction. We again start from the pieces
$
A_1,\dots,A_N
$
of upper density $1/N$, but now we add a separator set $Z$ with
$
\mu_{\rm up}(Z;K)=0
$
and define
$
L_j = Z \cup \bigcup_{i:\, j\in S_i} A_i.
$
The target is then enumerated in stages, alternating one point from each $A_i$ with a long block from $Z$. Once the finite bad sets are exhausted and the separator blocks are longer than the window, every late window either lies entirely inside $Z$ or intersects points from at most one set $A_i$. In the first case the generator is forced into
\[
K\cap \bigcap_{j=1}^n L_j = Z\,,
\]
which has upper density $0$. In the second case the generator is forced into
\[
K\cap \bigcap_{j\in S_i} L_j = A_i\cup Z\,,
\]
whose upper density is at most $1/N$. Thus no finite sliding window can beat the memoryless Sperner bound. The matching lower bound in \cref{thm:window-upper-density-exact} is immediate from \cref{thm:sperner-achievability}: the window-$W$ generator simply ignores the first $W-1$ entries of the window and applies the memoryless rule to the most recent example.

\paragraph{{Adaptive buffer}.} The proof of \cref{thm:no-eviction-universal-lower-bound} uses a greedy residual-version-space algorithm. Given a buffer state
$
M=(u_1,\dots,u_s),
$
we define the residual collection
$
\cL(M)\coloneqq \inbrace{L\in \cL : u_j\in L \text{ for every }j\in\insquare{s}}.
$
The update rule stores the current example exactly when the buffer is not full and the current example strictly shrinks this residual collection. The output rule is the canonical memoryless intersection generator applied to the current residual collection $\cL(M_{t-1})$.

Because each insertion removes at least one language from the residual collection and the buffer has size at most $b$, the buffer stabilizes after finitely many rounds at some final state $M_\ast$. If the final buffer is not full, then every language in $\cL(M_\ast)$ must contain the target $K$; otherwise the first witness in $K\setminus L$ would force one more insertion, contradicting stabilization. Since $K\in \cL(M_\ast)$ as well, the residual intersection is exactly $K$, and from that point onward the generator outputs $K$ itself, achieving upper density $1$.

If the final buffer is full, then each insertion has eliminated at least one language, so
$
\abs{\cL(M_\ast)}\leq k-b.
$
From the stabilization time onward, the learner is exactly the canonical memoryless intersection generator on this smaller residual family. Applying \cref{thm:sperner-achievability} to the residual collection yields the lower bound
\[
\frac{1}{\binom{k-b-1}{\floor{(k-b-1)/2}}}\,.
\]
Conceptually, each stored example removes one language from the effective ambiguity budget, and the memoryless Sperner bound is then applied to the reduced residual problem.

The formal details of this result appear in \cref{sec:density-bounds}.

\paragraph{Proof overview of \cref{thm:intro-incremental}.}
Recall that in this setting the learning algorithm cannot remember any of the
past examples, but it can remember its most recent output (in the form of an index $i_t$ of a language $L_{i_t} \in \cL.$).
{The goal is that, after some finite time step $t^\star$, the algorithm outputs only languages that have finite symmetric difference with the target $K$.}
We write  $A\preceq_{\rm F} B$ to denote that $A$ is an ``approximate'' subset of $B,$ \ie, it contains finitely many elements outside of $B$.
The first step of our
algorithm is to create a topological ordering of the
languages based on the $\preceq_{\rm F}$ ordering so that ``smaller'' languages appear earlier in the ordering.
(Since the collection is finite, such an ordering can be found.)
This prevents the overgeneralization
failure in which the learner gets stuck on a strict superset of $K$ and can never see evidence (in form of samples) that rules out this superset.

The algorithm now proceeds as follows: in the first round, it outputs the first language that is consistent with the input. In subsequent rounds, if the previously guessed language is consistent with the current input, the algorithm keeps the same guess. Otherwise, 
it moves one step forward in the topological ordering. 
To get the result, we first observe that 
if some language that the algorithm is currently guessing contains \emph{infinitely} many elements outside of $K,$
then after some finitely many timesteps the algorithm
will move on from this language as it will see an element that contradicts it. On the other hand, if the algorithm
is currently guessing $K$ it will never move on. Combining these two observations, along with some technical care, gives the formal result. The details appear in \cref{sec:incremental-bounded-memory}.

\subsection{Discussion and Open Questions}
	    In this work we {study variants} of memory-bounded
    language generation in the limit. {Our results show a rich landscape. When the learner is permitted to remember past examples but not its past outputs, generation in the limit is always possible, even under the most severe memory restriction. Obtaining non-trivial density bounds, however, forces us to focus on finite collections.}
	    {Conversely, when the learner forgets past examples but remembers only its last guess, approximate identification -- a mild relaxation of identification -- remains achievable for all finite collections, circumventing the classical impossibility \citep{lange1996incremental}.}
	    These results reinforce one of the main themes of the paper: {memory of past examples and memory of past outputs are different resources}, and bounded memory does not make all tasks uniformly harder. {Rather, the actual effect of these two resources is subtle and depends both on the learning objective and on the type of information retained.}
    
    There are several natural questions that arise from our work. The most immediate one concerns approximate identification with bounded memory. Without memory restrictions, approximate identification is achievable if and only if $\cL$ satisfies the weak Angluin condition \citep{charikar2025facets,kalavasis2026characterizations}. Does the same condition characterize approximate identification when the learner can only remember its previous guess, or does the memory restriction impose a stronger barrier?
    We also view our results as a step toward relaxing some of the unrealistic aspects of the language generation model, by showing that it is surprisingly robust to memory restrictions (\cref{thm:memoryless-finite-reps}). It would be interesting to understand which other natural restrictions on memory, computation, or access to past data still permit generation for all countable collections.

    \subsection{Related Work}
\label{sec:relatedworks}

Our work is most directly related to two lines of research: recent work on language generation in the limit, and the older literature on language identification in the limit under bounded-memory restrictions. The first line provides the generation frameworks that our paper builds on, while the second motivates the memory-restricted viewpoint that is central here.

\paragraph{Language generation in the limit.}
We work in the language generation in the limit framework introduced by \citet{kleinberg2024language}. {Since then, a growing body of work has explored many nearby questions and variants, including} uniform and non-uniform notions of generation, models with additional feedback, stronger notions of breadth or density, noisy or incomplete presentations, and partial identification \cite{li2025generation,kalavasis2025limits,charikar2025facets,raman2025noisy,peale2025representative,kleinberg2025density,kleinberg2025partial,hanneke2025union,mehrotra2025language,charikar2025characterization,karbasi2025impossibility,charikar2025pareto,arenas2025language,anastasopoulos2026safe,bai2026noise,kalavasis2026characterizations}. Our emphasis is different. Rather than strengthening the feedback model or studying broader, noisier, or more robust variants of generation, we focus on how much memory a learner must retain in order to generate.

\paragraph{Language generation in the limit with breadth or density.}
A particularly relevant subline of this literature studies stronger notions of generation that require the learner not just to produce valid unseen strings, but to do so with some notion of breadth or coverage. This direction is motivated by a basic feature of the original result of \citet{kleinberg2024language}: while their generator eventually stops producing strings outside the target language, it may do so by generating from an increasingly thin subset of the target. {Several subsequent works therefore introduced stronger requirements that ask the learner to generate more representative, denser, or otherwise broader subsets of the language. These strengthened notions turn out to be substantially harder than basic generation, in many cases bringing the task close to identification \cite{kalavasis2025limits,charikar2025facets,peale2025representative,kleinberg2025density,kalavasis2026characterizations,kleinberg2025partial}.} Our work is orthogonal to this direction. We keep the basic notion of generation fixed and instead study how the picture changes when the learner's memory is severely restricted.

\paragraph{Language identification in the limit with bounded memory.}
The closest prior line of work studies incremental, or iterative, learning in the limit, {where the learner forms its next guess using only its previous guess and the newest example} \cite{case1999incremental}. {In our framing, the learner forgets all past examples and retains only what is encoded in its current guess.} This line of work has grown into a broader literature on bounded-memory identification, studying variants such as bounded example memory and feedback learning \cite{case1999incremental}, limited long-term memory \cite{kinber1995language}, set-driven and rearrangement-independent restrictions \cite{lange1996set}, and temporary or ordinal-bounded memory \cite{lange2008incremental,carlucci2012learning}. 
{In contrast with the full-memory setting, prior work gives only limited positive results for identification under these bounded-memory models \cite{case1999incremental,kinber1995language,lange1996set}.} {Indeed, as we show in \cref{prop:incremental-not-identifiable}, there is a collection of just three languages that cannot be identified in this model. Surprisingly, however, this impossibility largely stems from the strictness of the definition: as \cref{thm:incremental-approximate-identification} shows, approximate identification, where the learner need only converge to a language differing from the target $K$ in finitely many elements, is achievable for every finite collection by an incremental learner.}

\paragraph{Language identification in the limit with bounded memory.}
{The incremental learning model has been extensively studied in the literature on language identification in the limit: here, the learner forms its next guess using only its previous guess and the newest example \cite{lange1996incremental,case1999incremental}.
Several variants of this model have also been studied.
These include bounded-memory-state learners, where the learner updates from the current example and an auxiliary state that can encode the previous guess together with other finite information \cite{koetzing2021memory}; counter variants, which keep the incremental learning model but also give the learner a numerical signal such as the current round number, in addition to the current example and previous guess \cite{case2008ushaped,koetzing2014iterative}; and models that retain past examples in restricted ways, including temporary example memory, where stored examples must eventually be forgotten if they are not seen again \cite{lange2008incremental}, ordinal-bounded example memory, where extensions of stored examples are controlled by an ordinal-valued budget \cite{carlucci2012learning}, and related restrictions such as limited long-term memory and set-driven or rearrangement-independent learning \cite{kinber1995language,lange1996set}.
Much of this literature compares these models to one another, establishing which are stronger than which and which are incomparable.
Our work asks a different question: how do bounded-memory restrictions affect the hierarchy of generation, density, and identification?
Our work asks a different question: how do bounded-memory restrictions affect the hierarchy of generation, density, and identification?
A central theme is that past examples and past outputs behave as different resources:
For a fixed amount of past-example memory, our density results show that the best minimax guarantee for finite collections deteriorates as the collection grows.
By contrast, in the last-guess setting, \cref{thm:incremental-approximate-identification} shows that every finite collection can be approximately identified; the eventual guess differs from the target on only finitely many elements, and therefore has density one in the target.}

\paragraph{Concurrent and independent work: space-efficient generation \citep{flammarion2026space}.}
{Concurrently and independently, \citet{flammarion2026space} also study language generation in the limit under resource constraints. Their work is incomparable with ours: They measure memory as the number of bits used by a streaming learner, and focus on target and hypothesis languages recognized by deterministic finite automata (DFA) with a bounded number of states over a finite alphabet. 
In that setting, they characterize tradeoffs between space used by the learner and the ``amount'' of the target hypothesis learned by the learner. We avoid automata-theoretic restrictions on the collection of languages, but model memory in a different and coarser way: we allow the learner to only retain a bounded number of past examples, or only its previous guess. Thus, the two papers ask the same broad question under different resource models.}

\section{Preliminaries}\label{sec:model-prelim}
 
In this section, we present the necessary background on language generation in the limit.
    
\paragraph{Notation.}
Let $\Sigma$ be a finite alphabet and let $\Sigma^\star$ denote the set of all finite strings over $\Sigma$.
We use $\cX$ to denote an arbitrary countable universe of strings; typically $\cX=\Sigma^\star$.
We fix an arbitrary canonical ordering of the domain $\cX$ and denote it by $(\bar{x}_1,\bar{x}_2,\dots).$
A language is an infinite subset $L\subseteq \cX$.
We denote a collection of languages using $\cL$ and restrict our attention to collections with finite or countably many elements. 
{For a set $A$, write $[A]^\infty$ for the collection of infinite subsets of $A$. For $W\in\N$, write $A^W$ for the set of ordered $W$-tuples of $A$ and $(A^W)_{\neq}\subseteq A^W$ for the subset of ones with distinct elements.}
{An enumeration} of a language $K$ is an infinite sequence $x_1,x_2,\ldots$ (possibly with repeats) such that each $x_t\in K$ and every element of $K$ appears at some time.
We use $S_n$ to denote the set of the first $n$ elements in the adversary's enumeration, $S_n \coloneqq \inbrace{x_1,\ldots,x_n}$.

\paragraph{Generating algorithms (with unbounded memory).}
A generating algorithm (or simply a generator) is a function $\generator$ that maps a {finite-length} history of observed examples from $\cX$ to a set of ``outputs.''
The generators are divided into different types {by output type}:
\begin{itemize}[leftmargin=15pt]
    \item \textbf{Element-Based Generators:}
        If $\generator$ outputs elements of $\cX$, {we say} it is an element-based generator.
    \item \textbf{Set-Based Generators:}
        If $\generator$ outputs \textit{infinite} sets, {we say} it is a set-based generator.
    \item \textbf{Index-Based Generators:}
        {Index-based generators are set-based generators that} \textit{always} output some language $L_i$ from the collection $\cL$.
        Often it is convenient to say index-based generators output an index $i\in \N${, meaning that they output} $L_i$.
\end{itemize}
Note that each of the above types of generators {retains} an arbitrarily large number of {examples it has observed} in the past. 
Later in this paper, we define more restricted types of generators with bounded memory. %
\begin{remark}[Types of Generators]
    {In the original formulation of \citet{kleinberg2024language}, the generator outputs a single string each round; later work} \cite{kleinberg2025density,li2025generation,kalavasis2026characterizations,charikar2025facets} {allows the generator to output a set} (or, equivalently, a sampling procedure) once sufficient training data have been observed.
\end{remark}

We now formally define language generation in the limit.
\begin{definition}[Language Generation in the Limit \citep{kleinberg2024language}]\label{def:consistentGeneration}
    Fix some $K$ from the language collection $\cL$ and a generating algorithm {$\generator=(\generator_n)$.}
    At each step $n$, let $S_n \subseteq K$ be the set of all strings that the algorithm $\generator$ has seen so far. 
    The {algorithm $\generator$} is said to generate from $K$ in the limit if, for all enumerations of $K$, there is some $n^\star \in \N$ such that for all steps $n \geq n^\star$, 
    \begin{enumerate}[itemsep=-2pt]
        \item if $\generator$ is element-based, then $\generator_n(S_n)\in K \setminus S_n$;
        \item if $\generator$ is set-based (including {index-based generators}), then $\generator_n(S_n)$ is infinite and $\generator_n(S_n)\subseteq K$.
    \end{enumerate}
    The collection $\cL$ allows for generation in the limit (or is generable) if there is an {algorithm $\generator$} that generates from $K$ in the limit for {every} $K \in \cL.$
\end{definition} 
{To gain} some intuition, consider the following example.
\begin{example}[Length-threshold languages]\label{ex:length-threshold}
Fix a finite alphabet $\Sigma$ and consider the countable collection of length-threshold languages $\cL=\inbrace{L_1,L_2,\ldots}$ where, for each $\ell$, $L_\ell \coloneqq \inbrace{x\in\Sigma^\star:\abs{x} \geq \ell}.$
Suppose the target is $K=L_{\ell^\star}$ and the adversary selects the enumeration $x_1,x_2,\ldots$.
After observing $S_n$, we know that $\ell^\star \leq \min_{x\in S_n}\abs{x}$.
Hence any string of length {at least} $m_n\coloneqq \min_{x\in S_n}\abs{x}$ belongs to every language consistent with $S_n$, and therefore belongs to $K$.
One valid element-based generator outputs the first element in the canonical enumeration of $\cX$ that has length at least $m_n$ and does not appear in $S_n$.
However, this is not a valid set-based generator because it outputs only a single element rather than an infinite set.
A valid set-based generator outputs all elements of $\cX$ of length at least $m_n$ that do not appear in $S_n$. 
\end{example} 

\section{Generators with No Memory}
\label{sec:memoryless:set}

In this section, we study generation when the generator carries no state from one round to the next and the current example is its only source of information.
{\memorylessgeneratordef*}
\noindent The generator is therefore not aware of the round number and retains no information about earlier examples. 
Depending on the output model, the output space $\Omega$ takes different forms: $\Omega=\cX$ for element-based generators, $\Omega=\cX^\infty$ for set-based generators (where $\cX^\infty$ denotes the set of all infinite subsets of $\cX$), and {for index-based generators, $\Omega=\inbrace{1,\ldots,N}$ when $\cL=\inbrace{L_1,\ldots,L_N}$ is finite, and $\Omega=\N$ when $\cL=\inbrace{L_1,L_2,\ldots}$ is countably infinite (where output $i$ is interpreted as the language $L_i$)}. 
A memoryless generator succeeds if its outputs satisfy the generation requirement of the corresponding output model (element, set, or index-based) on all sufficiently large rounds of every enumeration of the target language.

To build some intuition, fix a collection $\cL$ and a target language $K\in\cL$. Suppose that on round $t$ the adversary reveals $x_t\in K$ and asks the generator to output a fresh element of $K$, meaning one outside $\inbrace{x_1,\ldots,x_{t-1},x_t}$. A memoryless generator cannot in general do this: from the single input $x_t$, it has no way to tell which elements have already been shown. This is why the set-based model is the natural one in the memoryless setting. If the generator outputs an infinite subset of $K$, then because only finitely many examples have been shown so far, at least one element of this set is unseen.

{In the classical model of language identification in the limit \citep{gold1967language},} an enumeration of $K$ is any infinite sequence in which every element of $K$ appears at least once, and the adversary may repeat an element infinitely often. Under this regime, a memoryless generator is highly constrained. Once the current input $x$ is fixed, the generator must always produce the same response. For a set-based memoryless generator $\generator$ and a target $K$, call $x\in K$ bad if $\generator(x)\not\subseteq K$. If even one bad point exists, the adversary can repeat it every other round and use the remaining rounds to enumerate the rest of $K$.
A single mistake can therefore be amplified into infinitely many by the adversary.
This leads to the following natural restriction.

\begin{definition}[Finitely repeating enumerations]
\label{def:finitely-repeating}
An enumeration $(x_n)_{n\in\N}$ of an infinite language $K\subseteq \cX$ is finitely repeating if for every $x\in \cX$, the set $\inbrace{n\in\N : x_n=x}$ is finite.
\end{definition}
For convenience, we restate our main result below
\memorylessfiniterepsthm*
\noindent \cref{thm:memoryless-finite-reps} significantly generalizes the main result of \citet{kleinberg2024language}, which shows that arbitrary countable collections are generable in the limit when the generator has unrestricted memory. \cref{thm:memoryless-finite-reps} shows that the full countable-collection generality survives even after removing memory altogether, provided the adversary is subject to \cref{def:finitely-repeating}. The proof appears in \cref{sec:proofof:thm:memoryless-finite-reps}.

The next result shows that \cref{def:finitely-repeating} is essentially the weakest assumption one could hope for in the memoryless set-based model. Once arbitrary repetitions are allowed, {memoryless set-based generation is possible} only for very special collections.

\begin{theorem}[Characterization with arbitrary repetitions]
\label{thm:memoryless-with-reps}
Under arbitrary enumerations, a countable collection $\cL$ of infinite languages admits a memoryless set-based generator if and only if, in the unrestricted (full-memory) set-based model, every language in $\cL$ can already be generated after any single observed example.
Equivalently, for every $x\in \bigcup_{L\in\cL} L$, $\bigcap \inbrace{L\in\cL : x\in L}$ is infinite.
\end{theorem}
{The condition in \cref{thm:memoryless-with-reps} requires that each example $x$ by itself certify an infinite subset that is safe for every language containing it. Few natural collections satisfy this requirement. Thus \cref{def:finitely-repeating} marks the point at which the universal countable-collection guarantee in \cref{thm:memoryless-finite-reps} can hold in the memoryless model.} The proof of \cref{thm:memoryless-with-reps} appears in \cref{sec:proofof:thm:memoryless-with-reps}.

The positive result in \cref{thm:memoryless-finite-reps} also relies on set-based output in an essential way.
\begin{theorem}\label{thm:memoryless-set-based-necessary}
The set-based output in \cref{thm:memoryless-finite-reps} is necessary. Under finitely repeating enumerations, no infinite language is generable by a memoryless element-based generator, and there exists a collection of two infinite languages that is not generable by any memoryless index-based generator.
\end{theorem}
The proof of \cref{thm:memoryless-set-based-necessary} appears in \cref{sec:proofof:thm:memoryless-set-based-necessary}.

\subsection{Proof of \cref{thm:memoryless-finite-reps}}
\label{sec:proofof:thm:memoryless-finite-reps}

\begin{proof}
Since $\cX$ is countable, we may identify it with the positive integers without loss of generality. 
{Fix a sequence $L_1,L_2,\ldots$ whose range is $\cL$ and in which every language in $\cL$ appears at least once; if $\cL$ is finite, repeat its finite list to obtain such a sequence.} For $n,x\in\N$, let
\[
J_n(x)\coloneqq \bigcap \inbrace{L_j : j\leq n \text{ and } x\in L_j}\,,
\]
with the convention that the intersection of an empty family is $\N$. Then $J_1(x)$ is always infinite, so $n(x)\coloneqq \max \inbrace{n\leq x : J_n(x)\text{ is infinite}}$ is well-defined. Define the memoryless set-based generator by $\generator(x)\coloneqq J_{n(x)}(x)$.

Fix a target language $K=L_z$, and let $B_K\coloneqq \inbrace{x\in K : \generator(x)\not\subseteq K}$. To show that $\generator$ succeeds, it is enough to prove that $B_K$ is finite. Let $U_z$ be {the union of all sets $\bigcap \cF\coloneqq \cap_{L\in \cF}L$ such that $\cF\subseteq \inbrace{L_1,\ldots,L_z}$ and $\bigcap \cF$ is finite}; since there are only finitely many such subfamilies, $U_z$ is finite. Now take $x\in B_K$. If $n(x)\geq z$, then because $x\in L_z$, the language $L_z$ appears in the intersection defining $J_{n(x)}(x)$, and hence $\generator(x)=J_{n(x)}(x)\subseteq L_z=K$, a contradiction. Thus $n(x)<z$. If also $x\geq z$, then $J_z(x)$ must be finite, since otherwise $z$ would be admissible in the definition of $n(x)$; because $x\in L_z$, we have $x\in J_z(x)$, and therefore $x\in U_z$. Hence
\[
B_K \subseteq \inbrace{x\in\N : x<z}\cup U_z\,,
\]
so $B_K$ is finite.

If $E_K=\inparen{x_t}_{t\in\N}$ is a finitely repeating enumeration of $K$, then each point of the finite set $B_K$ appears only finitely many times in $E_K$. Therefore after some finite stage no bad point ever appears again, and from then on $\generator(x_t)\subseteq K$ at every round. Thus $\generator$ generates from $K$ in the limit.
\end{proof}

\subsection{Proof of \cref{thm:memoryless-with-reps}}
\label{sec:proofof:thm:memoryless-with-reps}

\begin{proof}
For each $x\in \bigcup_{L\in\cL} L$, write
\[
I_x \coloneqq \bigcap \inbrace{L\in\cL : x\in L}\,.
\]
We first note that the two formulations in the theorem are equivalent. Indeed, in the unrestricted set-based model, after observing the single example $x$, a generator succeeds simultaneously for every target language containing $x$ if and only if it can output some infinite set $A_x$ with $A_x\subseteq K$ for every $K\in\cL$ containing $x$. This is possible if and only if there exists an infinite subset of $I_x$, which is equivalent to $I_x$ itself being infinite. Thus it is enough to prove the theorem using the condition that $I_x$ is infinite for every $x$.

We now prove the two directions.
\begin{itemize}
    \item  \textbf{Necessity.}
Assume that $\cL$ admits a memoryless set-based generator $\generator$ under arbitrary enumerations. Suppose toward a contradiction that $I_x$ is finite for some $x\in \bigcup_{L\in\cL}L$. Since $\generator(x)$ must be infinite, we have $\generator(x)\not\subseteq I_x$. Hence there is some $K^\ast\in\cL$ with $x\in K^\ast$ but $\generator(x)\not\subseteq K^\ast$. Now let the adversary choose $K^\ast$ and enumerate it by repeating $x$ infinitely often, say as $x,y_1,x,y_2,\dots$, where $\inparen{y_i}_{i\in\N}$ lists the rest of $K^\ast$. Every time the current example is $x$, the generator outputs the same invalid set $\generator(x)\not\subseteq K^\ast$. Therefore it fails infinitely many times, contradicting generation in the limit.
\item  \textbf{Sufficiency.}
Now assume that $I_x$ is infinite for every $x\in \bigcup_{L\in\cL}L$. Define the memoryless generator by choosing $\generator(x)$ to be any infinite subset of $I_x$ for each such $x$; for $x\notin \bigcup_{L\in\cL}L$, define $\generator(x)$ arbitrarily. If the target language is $K\in\cL$ and the adversary presents an example $x_t\in K$, then by definition $I_{x_t}\subseteq K$, so $\generator(x_t)\subseteq K$. Hence the generator is correct on every round of every enumeration, and in particular generates from $K$ in the limit.
\end{itemize}
This proves the theorem.
\end{proof}

\subsection{Proof of \cref{thm:memoryless-set-based-necessary}}
\label{sec:proofof:thm:memoryless-set-based-necessary}

\begin{proof}
We treat the two output models separately.

\paragraph{Element-based generators.}
Let $\generator:\cX\to\cX$ be memoryless and fix any infinite language $K\subseteq \cX$. We show that some finitely repeating enumeration of $K$ makes $\generator$ fail infinitely often. Let $B\coloneqq \inbrace{x\in K : \generator(x)\notin K \text{ or } \generator(x)=x}$. If $B$ is infinite, choose distinct $b_1,b_2,\ldots\in B$ and a repetition-free enumeration $\inparen{z_i}_{i\in\N}$ of $K$; then $b_1,z_1,b_2,z_2,\ldots$ is a finitely repeating enumeration of $K$, and every round with current example $b_i$ is a failure. So assume $B$ is finite and set $G\coloneqq K\setminus B$. Then $G$ is infinite and $\generator(x)\in K\setminus\inbrace{x}$ for every $x\in G$. If some $y\in K$ has infinite fiber $\inbrace{x\in G : \generator(x)=y}$, choose distinct $a_1,a_2,\ldots$ from that fiber and a repetition-free enumeration $\inparen{z_i}_{i\in\N}$ of $K$; then $y,a_1,z_1,a_2,z_2,\ldots$ is finitely repeating, and every round with current example $a_i$ fails because $\generator(a_i)=y\in S_n$. Otherwise every fiber is finite, so $\generator(G)$ is infinite; choose distinct $a_1,a_2,\ldots\in G$ with distinct images $\generator(a_1),\generator(a_2),\ldots$, and let $\inparen{z_i}_{i\in\N}$ be a repetition-free enumeration of $K$. Then $\generator(a_1),a_1,z_1,\generator(a_2),a_2,z_2,\ldots$ is finitely repeating, and again every round with current example $a_i$ fails because $\generator(a_i)\in S_n$. Thus no infinite language is generable by a memoryless element-based generator.

\paragraph{Index-based generators.}
For the index-based part, it suffices to give a two-language counterexample. Let
\[
L_1\coloneqq \inbrace{n\in\N : n\equiv 0 \text{ or } 1 \pmod 4}
\qquad\text{and}\qquad
L_2\coloneqq \inbrace{n\in\N : n\equiv 0 \text{ or } 2 \pmod 4}\,.
\]
Their intersection $C\coloneqq L_1\cap L_2=\inbrace{n\in\N : n\equiv 0 \pmod 4}$ is infinite, while neither language contains the other. Let $\generator$ be any memoryless index-based generator for $\inbrace{L_1,L_2}$. On each $x\in C$, {the generator must output one of the two allowed hypotheses, $L_1$ or $L_2$}, so one of the sets $A_1\coloneqq \inbrace{x\in C : \generator(x)=L_1}$ and $A_2\coloneqq \inbrace{x\in C : \generator(x)=L_2}$ is infinite. If $A_1$ is infinite, choose distinct $a_1,a_2,\ldots\in A_1$ and a repetition-free enumeration $\inparen{z_i}_{i\in\N}$ of $L_2$; then $a_1,z_1,a_2,z_2,\ldots$ is a finitely repeating enumeration of $L_2$, and every round with current example $a_i$ fails because the generator outputs $L_1\not\subseteq L_2$. If $A_2$ is infinite, the symmetric construction with target $L_1$ gives infinitely many failures. Hence this two-language collection is not generable by any memoryless index-based generator.
\end{proof}

\section{Density Bounds for Generators with Limited Memory}
\label{sec:density-bounds}

{In this section we study density guarantees for generators with bounded memory. Generation in the limit asks only that outputs eventually lie in the target language; density additionally asks what fraction of the target language those outputs cover.}
We work throughout with set-based generators, and unless stated otherwise we continue to assume finitely repeating enumerations in the sense of \cref{def:finitely-repeating}.

The first question is what information about past examples the generator is allowed to retain.
{We compare three models: the memoryless model (which was also studied in the previous section), a sliding window of recent examples, and an adaptive buffer of examples chosen by the learner.}

\begin{definition}[Limited-memory set-based generators]
\label{def:limited-memory-generators}
A set-based generator is said to be of one of the following types.
\begin{enumerate}[itemsep=2pt,leftmargin=15pt]
    \item \textbf{Memoryless:} If it is memoryless in the sense of \cref{def:memoryless-generator}.

    \item \textbf{Sliding window of width $W$:} If, for each round $t$, the generator's output depends only on the block of the most recent $W$ examples seen so far: namely $x_1,\ldots,x_t$ when $t<W$, and $x_{t-W+1},\ldots,x_t$ when $t\geq W$.

    \item \textbf{Adaptive buffer of size $b$:} If the generator maintains a buffer $B_t\subseteq S_t$ with $\abs{B_t}\leq b$, initialized with $B_0=\emptyset$. On round $t$, both the output $G_t$ and the updated buffer $B_t$ depend only on the current example $x_t$ and the previous buffer $B_{t-1}$.
\end{enumerate}
\end{definition}
{Here, crucially, the sliding window keeps the most recent examples, informative or not, while the adaptive buffer  can keep an informative example once it appears, regardless of what arrives later.}
In particular, a sliding window of width $1$ is just the memoryless model, while the adaptive buffer of size $1$ is (weakly) more powerful than the memoryless model.

\paragraph{Density of outputs.}
{We use the same notions of density as \citet{kleinberg2025density}.} 
We compute all densities with respect to the fixed canonical ordering of $\cX$ from \cref{sec:model-prelim}.
Because every target language is infinite, a single output element has density $0$ in the target, so density questions are only meaningful for set-based outputs.

\begin{definition}[Upper and lower density]
\label{def:upper-lower-density}
Let $K\subseteq \cX$ be infinite, and list its elements in the canonical order inherited from $\cX$ as $K=\inbrace{\ell_1,\ell_2,\ldots}$. For $n\in\N$, write $K_{\leq n}\coloneqq \inbrace{\ell_1,\ldots,\ell_n}$. For any $S\subseteq \cX$, define
\[
\mu_{\rm up}(S;K)
~\coloneqq~ \limsup_{n\to\infty}~ \frac{\abs{S\cap K_{\leq n}}}{n}
\qquadand 
\mu_{\rm low}(S;K)
~\coloneqq~
\liminf_{n\to\infty}~ \frac{\abs{S\cap K_{\leq n}}}{n}\,.
\]
\end{definition}
We use $\limsup$ and $\liminf$ in the standard sense; in particular, $\mu_{\rm low}(S;K)\leq \mu_{\rm up}(S;K)$ for every $S$ and $K$. The two can nevertheless be as far apart as possible. For example, if $K=\N$ and
\[
S=\bigcup_{r\geq 1}\inbrace{n\in\N : (2r)!<n\leq (2r+1)!}\,,
\]
then $S$ consists of alternating blocks whose lengths grow very quickly. Along the subsequence $n=(2r)!$, only earlier blocks contribute, and their total size is negligible compared with $(2r)!$, so $\abs{S\cap K_{\leq n}}/n\to 0$. Along $n=(2r+1)!$, the omitted part consists only of earlier gaps, which are negligible compared with $(2r+1)!$, so the same ratio tends to $1$. Hence $\mu_{\rm low}(S;\N)=0$ and $\mu_{\rm up}(S;\N)=1$, so the gap between lower and upper density is as large as possible.

If a generator produces outputs $G_1,G_2,\ldots$ for a target language $K$, there are four natural ways to aggregate density over time:
\[
\liminf_{t\to\infty}\mu_{\rm low}(G_t;K)\,,\quad
\limsup_{t\to\infty}\mu_{\rm low}(G_t;K)\,,\quad
\liminf_{t\to\infty}\mu_{\rm up}(G_t;K)\,,\quadand 
\limsup_{t\to\infty}\mu_{\rm up}(G_t;K)\,.
\]
The $\liminf_t$ quantities require that outputs be dense from some point onward, whereas the $\limsup_t$ quantities only require that dense outputs appear infinitely often.
As we will show, three of these four notions collapse to zero in the bounded-memory models we consider, even for small finite collections. We therefore focus on the weakest:
\[
\limsup_{t\to\infty}\mu_{\rm up}(G_t;K)\,.
\]
Further, with bounded memory, there exist countable collections on which no generator achieves a positive $\limsup_t\mu_{\rm up}(G_t;K)$ for every target $K$. We therefore restrict attention to finite collections and ask for the best density guarantee a generator can achieve uniformly over all collections of a given size $k$.
For a particular collection $\cL$ of size $k$, we are allowed to tailor a generator to $\cL$; the adversary then chooses the hardest target language in $\cL$ and the hardest finitely repeating enumeration of it. %
The minimax upper density is the largest value that the generator can still guarantee in this game:
\begin{definition}[Minimax upper density]
\label{def:minimax-density-finite-collections}
Fix a memory model $\mathcal M$ from \cref{def:limited-memory-generators}, where for sliding-window and adaptive-buffer generators the parameter $W$ or $b$ is fixed as part of the model. For $k\geq 1$, define $\rho_{\rm up}^{\mathcal M}(k)$ to be the supremum over all $\sigma\in[0,1]$ with the following property:

For every collection $\cL$ of $k$ infinite languages, there exists a set-based generator $\generator$ in model $\mathcal M$ that generates from every language in $\cL$ in the limit and such that, for every target language $K\in \cL$ and every finitely repeating enumeration of $K$, if $G_1,G_2,\ldots$ are the outputs of $\generator$, then
\[
    \limsup_{t\to\infty}\mu_{\rm up}(G_t;K)\geq \sigma\,.
    \tag{Minimax Upper Density}
\]
\end{definition}
Equivalently, an upper bound $\rho_{\rm up}^{\mathcal M}(k)\leq \sigma$ means that there is some collection $\cL$ of $k$ infinite languages such that every generator in model $\mathcal M$ can be forced, by a suitable choice of target language and finitely repeating enumeration, to have
\[
\limsup_{t\to\infty}\mu_{\rm up}(G_t;K)\leq \sigma\,.
\]
For the three models above, we write $\rho_{\rm up}^{\rm mem}(k)$, $\rho_{\rm up}^{\rm win}(k,W)$, and $\rho_{\rm up}^{\rm buf}(k,b)$; when no ambiguity can arise, we abbreviate these to $\rho_{\rm up}(k)$, $\rho_{\rm up}(k,W)$, and $\rho_{\rm up}(k,b)$, respectively. 
The next sections determine these minimax densities for each of the three memory models.

\subsection{Density Bounds for Memoryless Generators}
{We begin with memoryless generators. With no memory of past examples, such a generator cannot distinguish target languages that share the current example. Therefore, except for finitely many inputs, its output on $x$ must lie inside every language in the collection that contains $x$.
The density is then controlled by how large these intersections can be inside the target language, and the answer is governed by the width of the Boolean lattice.}

\begin{theorem}[Memoryless Minimax Upper Density for Finite Collections]\label{thm:minimax-upper-density-memoryless-finite}
    In the memoryless generation setting, for every finite
    number $k \in \N$ it holds that
    \[
        \rho_{\rm up}(k) = \frac{1}{\binom{k-1}{\floor{ (k-1)/2}}}\,.
    \]
\end{theorem}
This result characterizes the minimax upper density in the memoryless setting. For lower density, we show the following strong negative result.

\begin{theorem}[{No uniform positive lower-density guarantee for size-$k$ collections}]
\label{thm:no-uniform-positive-density}
Fix any integer $k\geq 3$. There exists a size-$k$ collection $\cL$ of infinite languages
and a target $K\in \cL$ such that every memoryless set-based generator that generates
from all languages in $\cL$ in the limit (under {finitely repeating enumerations})
fails to achieve {lower} density $\sigma$ infinitely often on $K$ for every $\sigma>0$.
In particular,, for any repetition-free enumeration $E_K=\inparen{x_1,x_2,\dots}$ of $K$,
there is $t^\ast$ such that for all $t>t^\ast$,
\[
\mu_{\rm low}\!\inparen{\generator(x_t);K}=0\,.
\]
\end{theorem}
\noindent {Together, these two theorems separate the two density notions. Upper density admits a positive minimax guarantee for every finite collection size, although the guarantee decays quickly with $k$: at rate $\nfrac{\sqrt{k}}{2^{k-1}}$. Lower density admits no positive uniform guarantee once $k\geq 3$.}
The remainder of the section proves both theorems, with most of the work devoted to the upper-density minimax.

    \subsubsection{{No Positive Lower-Density Guarantee}}   
   {We first prove \cref{thm:no-uniform-positive-density}. The proof uses a simple partition argument.}

\begin{definition}[{Partition inside a collection}]
    {A language $K \in \cL$ is partitioned inside $\cL$ by $L_1,\dots,L_k$ if $L_1,\dots,L_k\in\cL$ are pairwise disjoint and $K=\bigcup_{i=1}^k L_i$.}
\end{definition}

\begin{lemma}[{Lower-density bound from partitions}]
\label{thm:densityCollapse}
    Let $\cL$ be a countable collection and $\generator: \cX \to [\cX]^\infty$ be a memoryless set-based generator that generates from every $L \in \cL$ in the limit under repetition-free enumerations. If $K \in \cL$ is partitioned by $L_1,\dots,L_k\in \cL$, then for any repetition-free enumeration of $K$, there exists a time $t^\star$ such that for all $t > t^\star$:
    \[
        \mu_{\rm low}(\generator(x_t), K) ~\leq~ \max\nolimits_{i \in [k]} \mu_{\rm low}(L_i, K)\,.
    \]
    In particular, if the partition is uniform (\ie, $\mu_{\rm low}(L_i, K) = \nfrac{1}{k}$ for all $i$), the generator cannot infinitely often output a set with density greater than $\nfrac{1}{k}$.
\end{lemma}

\begin{proof}
    For each $i \in \{1, \dots, k\}$, let $B_i \coloneqq \{x \in L_i \mid \generator(x) \not\subseteq L_i\}$ be the set of elements where $\generator$ fails to be consistent with language $L_i$. By the success condition of \cref{def:consistentGeneration}, each $B_i$ must be finite{: if some $B_i$ were infinite, any repetition-free enumeration of $L_i$ would encounter infinitely many points from $B_i$, contradicting generation in the limit}.
    
    Define the ``total error budget'' $B \coloneqq \bigcup_{i=1}^k B_i$. Since $B$ is a finite union of finite sets, $B$ is finite. Now, consider any $x \in K \setminus B$.
    Because $\{L_i\}_{i=1}^k$ partitions $K$, there exists a unique $j \in \{1, \dots, k\}$ such that $x \in L_j$. Since $x \notin B$, it follows that $x \notin B_j$, which implies:
    \[
        \generator(x) \subseteq L_j\,.
    \]
    {By the monotonicity of the density and $\generator(x)\subseteq L_j$,} $\mu_{\rm low}(\generator(x), K) \leq \mu_{\rm low}(L_j, K) \leq \max_i \mu_{\rm low}(L_i, K)$.
    
    Finally, since the enumeration $E_K = (x_1, x_2, \dots)$ is repetition-free, the finite set $B$ can only appear in the stream for a finite number of rounds. Let $t^\star = \max \{t \mid x_t \in B\}$. For all $t > t^\star$, the example $x_t$ is in $K \setminus B$, and the density bound holds.
\end{proof}
    Next we need a partition lemma for countable sets.

        \begin{lemma}[Zero-lower-density partitions of countable sets]
\label{lem:zero-density-partition}
Let $K$ be any countably infinite set with a canonical enumeration
$K=\inbrace{x_1,x_2,\dots}$, and let $m\geq 2$ be an integer.
Then there exist pairwise disjoint infinite subsets $A_1,\dots,A_m\subseteq K$
with $\bigcup_{i=1}^m A_i = K$ such that for every $i\in\insquare{m}$,
\[
\mu_{\rm low}(A_i;K)=0\,,
\qquad\text{where}\qquad
\mu_{\rm low}(A_i;K)\coloneqq \liminf_{n\to\infty}\frac{\abs{A_i\cap K^n}}{n}\,,
\ \ K^n\coloneqq \inbrace{x_1,\dots,x_n}\,.
\]
\end{lemma}

\begin{proof}
We define a set of ``rapidly'' growing blocks.
Let $s_0\coloneqq 0$. For each $t\geq 1$, define
\[
\ell_t \coloneqq t^2\inparen{1+s_{t-1}}\,,
\qquad
s_t \coloneqq s_{t-1}+\ell_t\,,
\qquad
B_t \coloneqq \inbrace{s_{t-1}+1,\dots,s_t}\subseteq \N\,.
\]
Then $\inbrace{B_t}_{t\geq 1}$ is a partition of $\N$ into consecutive nonempty blocks, where the $t$-th block has length $\ell_t$.
Observe that:
\begin{equation}
\label{eq:st-product}
1+s_t = \inparen{1+t^2}\inparen{1+s_{t-1}}
\quad\text{for all }t\geq 1\,,
\end{equation}
Next, for each $i\in\insquare{m}$ let
\[
I_i\coloneqq \inbrace{t\in\N: t\equiv i \pmod m}
\qquadand 
A_i \coloneqq \bigcup_{t\in I_i} \inbrace{x_j : j\in B_t}\,.
\]
Because the blocks $B_t$ partition $\N$, the sets $A_1,\dots,A_m$ are pairwise disjoint and satisfy
$\bigcup_{i=1}^m A_i
=
K.$
Moreover each $A_i$ is infinite. 
It remains to show that for each 
$i\in\insquare{m}$, $\liminf_{n\to\infty}\frac{\abs{A_i\cap K^n}}{n}=0.$
Fix any $i\in\insquare{m}$.
Consider any block index $t\geq 2$ with $t\notin I_i$ (there are infinitely many such $t$ because $I_i$
contains exactly one residue class modulo $m$).
Set $n\coloneqq s_t$. Then $K^{s_t}=\inbrace{x_1,\dots,x_{s_t}}$ consists of exactly those
enumerated elements whose indices lie in $\bigcup_{r\leq t} B_r$.
Since $t\notin I_i$, among the blocks $\inbrace{B_1,\dots,B_t}$, the ones assigned to bin $i$
are contained in $\inbrace{B_1,\dots,B_{t-1}}$. Therefore
\begin{align*}
\abs{A_i\cap K^{s_t}}
&=
\abs{\Bigl(\bigcup\nolimits_{r\in I_i} \inbrace{x_j:j\in B_r}\Bigr)\cap \inbrace{x_1,\dots,x_{s_t}}}\\
&=
\abs{\bigcup\nolimits_{\substack{r\in I_i,~ r\leq t}} \inbrace{x_j:j\in B_r}}\\
&\leq
\abs{\bigcup\nolimits_{r\leq t-1} \inbrace{x_j:j\in B_r}}\\
&=
s_{t-1}\,.
\end{align*}
Dividing by $s_t$ gives
\begin{equation}
\label{eq:ratio-bound}
\frac{\abs{A_i\cap K^{s_t}}}{s_t}
\leq
\frac{s_{t-1}}{s_t}
=
\frac{s_{t-1}}{s_{t-1}+\ell_t}
=
\frac{s_{t-1}}{s_{t-1}+t^2\inparen{1+s_{t-1}}}
\leq
\frac{s_{t-1}}{\inparen{1+t^2}s_{t-1}}
=
\frac{1}{1+t^2}\,,
\end{equation}
where we used $s_{t-1}\geq 1$ for $t\geq 2$ and $\ell_t=t^2\inparen{1+s_{t-1}}$.
Since there are infinitely many $t\notin I_i$, inequality~\eqref{eq:ratio-bound} holds along an infinite subsequence
$n=s_t\to\infty$, and thus
\[
\liminf_{n\to\infty}\frac{\abs{A_i\cap K^n}}{n}
~\leq~
\lim_{t\to\infty,\, t\notin I_i}\frac{\abs{A_i\cap K^{s_t}}}{s_t}
~\leq~
\lim_{t\to\infty}\frac{1}{1+t^2}
~=~
0\,.
\]
Because the ratio $\abs{A_i\cap K^n}/n$ is always nonnegative, the $\liminf$ cannot be negative, hence it equals $0$.
This shows $\mu_{\rm low}(A_i;K)=0$ for every $1\leq i\leq m$.
\end{proof}
Having shown \cref{lem:zero-density-partition}, we are now ready to prove \cref{thm:no-uniform-positive-density}.

\begin{proof}[Proof of \cref{thm:no-uniform-positive-density}]
Fix {$k\geq 3$} and let $K$ be any countably infinite language (for concreteness, $K=\N$).
Apply \cref{lem:zero-density-partition} with $m=k-1$ to obtain a partition
$K=A_1\cup\cdots\cup A_{k-1}$ into infinite sets with $\mu_{\rm low}(A_i;K)=0$ for all $i$.
Let
\[
\cL \coloneqq \inbrace{K, A_1,\dots,A_{k-1}}\,.
\]
Now let $\generator$ be any memoryless set-based generator that generates from every language in $\cL$
in the limit under {finitely repeating} enumerations.
By \cref{thm:densityCollapse} (applied to {the partition of $K$ by $A_1,\dots,A_{k-1}$ inside $\cL$}),
for any repetition-free enumeration of $K$ there exists $t^\ast$ such that for all $t>t^\ast$,
\[
\mu_{\rm low}\!\inparen{\generator(x_t);K}
~\leq~
\max_{i\in\insquare{k-1}} \mu_{\rm low}(A_i;K)
~=~
0\,.
\]
Hence $\mu_{\rm low}\!\inparen{\generator(x_t);K}=0$ for all sufficiently large $t$, so $\generator$ cannot
achieve density $\sigma$ infinitely often on $K$ for any $\sigma>0$.
\end{proof}

\subsubsection{Minimax Upper Set-Density for Finite Collections}

To establish the sharp density bounds for memoryless generators on finite collections, we rely on two foundational results from order theory regarding the Boolean lattice of subsets.

\begin{fact}[Sperner's Theorem]\label{fact:sperner}
    Let $S$ be a finite set of size $n$, and let $\cP(S)$ be its power set. An antichain is a subcollection $\cA \subseteq \cP(S)$ such that for any $A, B \in \cA$ with $A \neq B$, neither $A \subseteq B$ nor $B \subseteq A$ holds. The maximum possible size of an antichain in $\cP(S)$ is exactly $\binom{n}{\floor{n/2}}$.
\end{fact}

\begin{fact}[Symmetric Chain Decomposition \citep{engel1997sperner}]\label{fact:symmetric-chain}
    Let $S$ be a finite set of size $n$. A chain in $\cP(S)$ is a sequence of strictly increasing subsets $C_1 \subset C_2 \subset \dots \subset C_m$. The Boolean lattice $\cP(S)$ can be partitioned into exactly $\binom{n}{\floor{n/2}}$ disjoint symmetric chains, such that every subset in $\cP(S)$ belongs to exactly one chain.
\end{fact}
\noindent We now show that for memoryless set-based generators, the achievable upper density on finite collections is entirely governed by the Sperner bound. We first prove that an adversary can always force the upper density to drop to this bound.

\begin{lemma}[Density Bound for Memoryless Generators]\label{thm:upper-density-collapse}
    Fix any integer $k \geq 2$. There exists a collection $\cL$ of size $k$ and a target language $K \in \cL$ such that the following holds: for every memoryless set-based generator $\generator$ that generates from $\cL$ in the limit (under finitely repeating enumerations), and for any repetition-free enumeration $E_K = (x_1, x_2, \dots)$ of $K$, there exists a time $t^\ast$ such that for all $t > t^\ast$:
    \[
        \mu_{\rm up}\inparen{\generator(x_t); K} \leq \frac{1}{\binom{k-1}{\floor{ (k-1)/2}}}\,.
    \]
\end{lemma}

\begin{proof}
    For $k = 2$, the upper density bound evaluates to $1 / \binom{1}{0} = 1$. Because the upper density of any set cannot exceed 1, the theorem holds trivially for any target in any collection of size 2. Therefore, we may assume without loss of generality that $k \geq 3$. 
    
    Let $n = k - 1 \geq 2$, and define $N \coloneqq \binom{n}{\floor{n/2}}$. 
    Let the target language $K$ be a countably infinite universe. We partition $K$ into $N$ pairwise disjoint infinite sets, $A_1, A_2, \dots, A_N$, such that the asymptotic upper density of each set within $K$ is exactly uniform, \ie{}, for each $i \in [N]$:
        $\mu_{\rm up}(A_i; K) = \frac{1}{N}\,.$
    (Such a partition can be constructed easily, for instance, by assigning elements in a round-robin fashion according to $K$'s canonical enumeration, meaning the $m$-th element belongs to partition $A_i$ where $i \equiv m \pmod N$).

    Let $S_1, S_2, \dots, S_N$ be the $N$ distinct subsets of $\inbrace{1, 2, \dots, n}$ that have a cardinality of exactly $\floor{n/2}$. By \cref{fact:sperner}, this collection of subsets forms an antichain. Furthermore, since $n \geq 2$, we have $\floor{n/2} \geq 1$, which guarantees that each subset $S_i$ is non-empty.

    We now construct the remaining $n$ languages in our collection. For each $j \in [n]$, define:
    \[
        L_j \coloneqq \bigcup_{i:\, j \in S_i} A_i\,.
    \]
Set $r \coloneqq \floor{n/2}$. Since we have already reduced to the case
$k \geq 3$, we have $n \geq 2$ and hence $r \geq 1$.
Fix any $j \in [n]$. Among the $r$-subsets of $[n]$, the number
that contain $j$ is
\[
\binom{n-1}{r-1} > 0\,.
\]
Because $S_1,\dots,S_N$ enumerate all $r$-subsets of $[n]$, there exists
some $i \in [N]$ such that $j \in S_i$. For this index $i$, we have
\[
A_i \subseteq L_j\,.
\]
Since $A_i$ is infinite, it follows that $L_j$ is infinite. Moreover, by construction, the $L_j$'s are  distinct.
Let our finite collection be $\cL \coloneqq \inbrace{K, L_1, \dots, L_n}$. Notice that
$\abs{\cL} = n+1 = k$.

    Let $\generator$ be any memoryless set-based generator that successfully generates from $\cL$. By the definition of generation in the limit under finitely repeating enumerations, for every $L \in \cL$, the set of errors $B_L \coloneqq \inbrace{x \in L \mid \generator(x) \not\subseteq L}$ is finite{: otherwise any repetition-free enumeration of $L$ would induce infinitely many errors}. Because $\cL$ is a finite collection, the ``total error budget'' $B \coloneqq \bigcup_{L \in \cL} B_L$ is also finite.

    Consider any element $x \in K \setminus B$. Because $A_1, \dots, A_N$ partition $K$, there is a unique index $i \in [N]$ such that $x \in A_i$. By our language construction, $x \in L_j$ if and only if $A_i \subseteq L_j$, which occurs if and only if $j \in S_i$. 

    Because $x \notin B$, the generator must output a set consistent with all languages containing $x$. Thus:
    \[
        \generator(x) \subseteq K \cap \bigcap_{j \in S_i} L_j\,.
    \]
    Expanding the intersection, we have:
    \[
        \bigcap_{j \in S_i} L_j = \bigcap_{j \in S_i} \inparen{ \bigcup_{r:\, j \in S_r} A_r }\,.
    \]
    Because the sets $A_r$ are pairwise disjoint, an element in this intersection must belong to a single partition $A_r$ that is present in the union for \emph{every} $j \in S_i$. This requires $j \in S_r$ for all $j \in S_i$, meaning $S_i \subseteq S_r$. Because $S_1, \dots, S_N$ form an antichain, $S_i \subseteq S_r$ implies $S_i = S_r$, and thus $i = r$. 
    Therefore, the intersection reduces to $A_i$. For any $x \in K \setminus B$ where $x \in A_i$, we have $\generator(x) \subseteq A_i$.

    Finally, let $E_K = (x_1, x_2, \dots)$ be a repetition-free enumeration of $K$. Because $B$ is finite, there exists $t^\ast$ such that for all $t > t^\ast$, $x_t \notin B$. For all such $t$, the generator is forced to output $\generator(x_t) \subseteq A_i$ for some $i$. {Therefore,} by the monotonicity of the upper density operator
    \[
        \mu_{\rm up}\inparen{\generator(x_t); K} \leq \mu_{\rm up}(A_i; K) = \frac{1}{N}\,.
    \]
    Substituting $N = \binom{k-1}{\floor{ (k-1)/2}}$ completes the proof.
\end{proof}
We now show that this bound is tight. {Using a symmetric chain decomposition, a generator can always guarantee this upper density infinitely often.}

\begin{lemma}[Achievability of the Sperner Bound]\label{thm:sperner-achievability}
    For any collection $\cL$ of $k \geq 2$ infinite languages, there exists a memoryless set-based generator that guarantees an upper density of at least $1 / \binom{k-1}{\floor{ (k-1)/2}}$ infinitely often on any target $K \in \cL$ {and any finitely repeating enumeration of $K$}.
\end{lemma}

\begin{proof}
    We first define a modified canonical memoryless set-based generator. For any $x \in \cX$, define its signature as $S(x) \coloneqq \inbrace{L \in \cL \mid x \in L}$, and let $I_x \coloneqq \bigcap_{L \in S(x)} L$ be the intersection of all languages in $\cL$ that contain $x$. If $I_x$ is infinite, we set $\generator(x) \coloneqq I_x$. If $I_x$ is finite, we set $\generator(x) \coloneqq \cX$ (or any arbitrary infinite subset of $\cX$) to ensure the output remains an infinite set.

    Let us verify that this generator successfully generates from $\cL$ in the limit. For any target $L \in \cL$, an error ($\generator(x) \not\subseteq L$) can only occur if the finite-intersection fallback is triggered. For any subcollection $B \subseteq \cL$, define the exact region $R_B \coloneqq \inbrace{x \in \cX \mid S(x) = B}$. By definition, any $x \in R_B$ satisfies $x \in I_B \coloneqq \bigcap_{L \in B} L$, meaning $R_B \subseteq I_B$. Thus, if a signature $B$ yields a finite intersection $I_B$, its corresponding exact region $R_B$ must also be finite. Since a finite collection $\cL$ has only finitely many distinct signatures, the union of all exact regions corresponding to finite intersections is a finite set. Therefore, the generator makes only finitely many errors globally, {and each of these error-causing inputs can appear only finitely many times under finitely repeating enumerations}.

    Now, fix a target language $K \in \cL$. There are $n = k - 1$ other languages in the collection. For any element $x \in K$, we define its relative signature restricted to the other languages: $S_K(x) \coloneqq \inbrace{L \in \cL \setminus \inbrace{K} \mid x \in L}$. 

    For every subset $B \subseteq \cL \setminus \inbrace{K}$, define the relative region $R'_B \coloneqq \inbrace{x \in K \mid S_K(x) = B}$. These regions are pairwise disjoint and partition $K$. 

    For any $x \in R'_B$, the intersection of all consistent languages in $\cL$ (including $K$) is exactly $I'_B \coloneqq K \cap \bigcap_{L \in B} L$. Because any element $y \in I'_B$ belongs to $K$ and to all languages in $B$, its relative signature $S_K(y)$ must be a superset of $B$. Conversely, any element in $K$ with a relative signature $C \supseteq B$ clearly belongs to $I'_B$. Therefore, we have the equality:
    \[
        \bigcup_{C \supseteq B} R'_C = I'_B\,.
    \]
    By \cref{fact:symmetric-chain}, the Boolean lattice of subsets of $\cL \setminus \inbrace{K}$ can be partitioned into exactly $N = \binom{n}{\floor{n/2}}$ chains, which we denote $\cC_1, \cC_2, \dots, \cC_N$. 

    Because these chains partition the power set of $\cL \setminus \inbrace{K}$, their corresponding families of relative regions partition $K$. Let $U_m$ be the union of all relative regions associated with the subsets in chain $\cC_m$:
    \[
        U_m \coloneqq \bigcup_{B \in \cC_m} R'_B\,.
    \]
    Since $\bigcup_{m=1}^N U_m = K$, and the limit superior is finitely subadditive, we have:
    \[
        1 = \mu_{\rm up}(K; K) \leq \sum_{m=1}^N \mu_{\rm up}(U_m; K)\,.
    \]
    By the Pigeonhole Principle, there must exist at least one chain $\cC^\ast$ such that:
    \[
        \mu_{\rm up}(U^\ast; K) \geq \frac{1}{N}\,.
    \]
    Let the chain $\cC^\ast$ be defined by the sequence of subsets $B_1 \subset B_2 \subset \dots \subset B_h$. Let $B_j$ be the smallest subset in this chain such that its relative region $R'_{B_j}$ is infinite. (Such a $B_j$ must exist, because if all regions in the chain were finite, $U^\ast$ would be finite, yielding an upper density of $0$, which contradicts $1/N > 0$).

    Because the regions $R'_{B_1}, \dots, R'_{B_{j-1}}$ are finite, their upper density is exactly $0$. By the finite subadditivity of the upper density, dropping these finite regions does not decrease the density of the rest of the chain. Thus, the upper density of the entire chain $U^\ast$ is carried exclusively by the remaining regions:
    \[
        \mu_{\rm up}\inparen{ \bigcup\nolimits_{i \geq j} R'_{B_i} ; K } \geq \mu_{\rm up}(U^\ast; K) \geq \frac{1}{N}\,.
    \]
    Because $B_j \subseteq B_i$ for all $i \geq j$, the union of these is precisely a subset of the exact intersection $I'_{B_j}$:
    \[
        \bigcup_{i \geq j} R'_{B_i} \subseteq \bigcup_{C \supseteq B_j} R'_C = I'_{B_j}\,.
    \]
    By monotonicity, the upper density of the intersection is at least $1/N$:
    \[
        \mu_{\rm up}(I'_{B_j}; K) \geq \frac{1}{N}\,.
    \]  
    Finally, because $\mu_{\rm up}(I'_{B_j}; K) \geq 1/N > 0$, the intersection $I'_{B_j}$ is infinite. Thus, for any element $x_t \in R'_{B_j}$, {the finite-intersection fallback is not triggered; it outputs exactly} the infinite set $I'_{B_j}$ (since $I_{x_t} = I'_{B_j}$). 

    Furthermore, because the relative region $R'_{B_j}$ is infinite, any valid {finitely repeating} enumeration of $K$ must present elements $x_t \in R'_{B_j}$ infinitely many times. At each of these infinitely many time steps, the generator outputs $I'_{B_j}$, {achieving} an upper density of at least $1/N$.
\end{proof}
Equipped with the previous results, the proof of \cref{thm:minimax-upper-density-memoryless-finite} follows as a direct corollary:
\begin{proof}[Proof of \cref{thm:minimax-upper-density-memoryless-finite}]
    {If $k=1$, then for the unique target language $K$ the generator can output $K$ on every round, so the minimax value is $1=1/\binom{0}{0}$. Hence assume $k\geq 2$.}
    From \cref{thm:upper-density-collapse}, we get that
    $\rho_{\rm up}(k) \leq \frac{1}{\binom{k-1}{\floor{ (k-1)/2}}}.$ Then, \cref{thm:sperner-achievability} shows that $\rho_{\rm up}(k) \geq \frac{1}{\binom{k-1}{\floor{ (k-1)/2}}}.$ This concludes the proof.
\end{proof}

\begin{remark}[Density Bounds for Countable Collections]
 {The finite-collection bounds also imply that no positive uniform upper-density guarantee can hold over all countable collections. One way to see this is to take a disjoint union of the finite hard collections for all $k\in\N$; any positive guarantee would contradict the finite bound for sufficiently large $k$.}
\end{remark}

\subsection{{Density Bounds for Generators with a Sliding Window}}
\label{sec:sliding-window-generators}

{We next study the sliding-window model, in which the generator can inspect the last $W$ examples in the stream. Throughout this subsection, we restrict attention to repetition-free enumerations. We first give the formal definition of generation in the limit in this model.}

\begin{definition}[Window-$W$ set-based generator]
\label{def:window-generator}
Fix an integer $W\geq 1$. A \emph{window-$W$ set-based generator} is a deterministic function
\[
\generator:(\cX^W)_{\neq} \to [\cX]^\infty\,,
\qquadwhere
(\cX^W)_{\neq}
\coloneqq
\inbrace{(y_1,\dots,y_W)\in \cX^W : y_a\neq y_b \text{ for all } a\neq b}\,.
\]
{Here, $(\cX^W)_{\neq}$ the set of ordered $W$-tuples of distinct elements of $\cX$, and $[\cX]^\infty$ is the family of infinite subsets of $\cX$.}
Given a repetition-free enumeration $E_K=(x_1,x_2,\dots)$ of a target language $K$, the output at time $t\geq W$ is
\[
G_t \coloneqq \generator(x_{t-W+1},\dots,x_t)\,.
\]
{Note that, the first $W-1$ rounds may be assigned arbitrary infinite outputs, since all guarantees below are eventual.}
We say that $\generator$ \emph{generates from $K$ in the limit under repetition-free enumerations} if for every repetition-free enumeration $E_K$ there exists $t^\ast$ such that for all $t\geq \max\!\inbrace{W,t^\ast}$, we have $G_t\subseteq K$. We say that $\generator$ generates in the limit from a collection of languages $\cL$
if it generates in the limit from every $K \in \cL.$
\end{definition}
{We ask the natural sliding-window analogue of the memoryless minimax question: how large an upper-density guarantee can a window-$W$ generator achieve uniformly over all collections of size $k$? 
We show that a longer window does not improve over the memoryless setting for any finite $W$.}

{Recall that $\rho_{\rm up}(k,W)$ denotes the size-$k$ minimax upper-density value in the window-$W$ model (the natural analogue of \cref{def:minimax-density-finite-collections}).}
Our main result in this section is as follows:

\begin{theorem}[Minimax upper-density value for sliding windows]
\label{thm:window-upper-density-exact}
For every $k\geq 1$ and $W\geq 1$,
\[
\rho_{\rm up}(k,W)
=
\frac{1}{\binom{k-1}{\floor{(k-1)/2}}}\,.
\]
\end{theorem}
In particular, increasing the window length does not improve the worst-case upper-density guarantee for finite collections.

{The rest of this subsection proves the \cref{thm:window-upper-density-exact}. The first step is the following lemma, which is a sliding-window analogue of a basic observation from the memoryless setting.}

\begin{lemma}[{Finite exceptional set for bad windows}]
\label{lem:finite-bad-window-budget}
Let $W\geq 1$, let $L\subseteq \cX$ be an infinite language, and let $\generator:(\cX^W)_{\neq}\to[\cX]^\infty$ be a window-$W$ set-based generator that generates from $L$ in the limit under repetition-free enumerations. Then there exists a finite set $B_L\subseteq L$ such that for every ordered $W$-tuple $(y_1,\dots,y_W)$ of distinct elements of $L\setminus B_L$,
\[
\generator(y_1,\dots,y_W)\subseteq L\,.
\]
\end{lemma}

\begin{proof}
Assume for contradiction that no such finite set $B_L$ exists.
Then for every finite set $F\subseteq L$ there exists an ordered $W$-tuple of distinct elements
\[
(y_1^F,\dots,y_W^F)\in {(L\setminus F)^W_{\neq}}
\qquadtext{for which}
\generator(y_1^F,\dots,y_W^F)\not\subseteq L\,.
\]
Fix any repetition-free enumeration $z_1,z_2,\dots$ of $L$.
We construct another repetition-free enumeration of $L$ in stages.
Let $U_0\coloneqq\emptyset$.
At stage $s\geq 1$, let $u_s$ be the first element of the sequence $z_1,z_2,\dots$ that does not belong to $U_{s-1}$, and append $u_s$ to the output sequence.
Now apply the assumption above with
$F\coloneqq U_{s-1}\cup\inbrace{u_s}$
to obtain an ordered $W$-tuple
$\tau_s=(y_{s,1},\dots,y_{s,W})$
of distinct elements of $L\setminus F$ such that
$\generator(y_{s,1},\dots,y_{s,W})\not\subseteq L.$
Append the $W$ entries of $\tau_s$ immediately after $u_s$, and let $U_s$ be the set of all elements that have been written so far.

By construction, every newly appended element lies outside the previously used set, so the resulting sequence is repetition-free. It is also an enumeration of $L$: each $z_r$ either appears earlier inside one of the tuples $\tau_s$, or else, once all $z_1,\dots,z_{r-1}$ have been used, it becomes the selected element $u_s$ at some later stage. Thus every element of $L$ appears exactly once.

Finally, at the last position of the block $\tau_s$, the current window is exactly
\[
(y_{s,1},\dots,y_{s,W})\,,
\]
so the generator outputs a set not contained in $L$. This happens for every stage $s$, contradicting the assumption that $\generator$ generates from $L$ in the limit.
\end{proof}
The next result is the main technical component needed for \cref{thm:window-upper-density-exact}. In fact, it shows something stronger than what is needed for the minimax bound: there is a \emph{single} collection $\cL$ and target $K \in \cL$, together with a \emph{single} adversarial repetition-free enumeration of $K$, that witnesses the upper bound for \emph{every} finite window length $W$.

\begin{lemma}[A single hard instance for all finite sliding windows]
\label{lem:window-upper-density-collapse}
Fix an integer $k\geq 2$. There exists a collection $\cL$ of size $k$, a target language $K\in \cL$, and a fixed repetition-free enumeration $E_K$ of $K$ such that the following holds. For every integer $W\geq 1$ and every window-$W$ set-based generator $\generator$ that generates from $\cL$ in the limit under repetition-free enumerations, there exists a time $t^\ast$ such that for all $t\geq t^\ast$,
\[
\mu_{\rm up}(G_t;K)
\leq
\frac{1}{\binom{k-1}{\floor{(k-1)/2}}}\,,
\]
where $G_t$ denotes the output of $\generator$ on the fixed enumeration $E_K$ at time $t$.
\end{lemma}

\begin{proof}
For $k=2$, the claimed bound is $1$, and the statement is trivial since every upper density is at most $1$.
Hence we may assume $k\geq 3$.
Set
\[
n\coloneqq k-1
\qquadand 
N\coloneqq \binom{n}{\floor{n/2}}\,.
\]
Let $S_1,\dots,S_N$ be the $N$ subsets of $[n]$ of size exactly $\floor{n/2}$.
Thus $S_1,\dots,S_N$ form the middle layer of the Boolean lattice on $[n]$.

Choose any countably infinite subset of the domain $K$, and write its canonical ordering as $K=\inbrace{x_1,x_2,\dots}.$
We partition $K$ into pairwise disjoint infinite sets
\[
K=A_1 \cup \cdots \cup A_N \cup Z
\]
with the properties
\[
\mu_{\rm up}(A_i;K)=\frac{1}{N}
\quad\text{for every }i\in[N]\,,
\qquad\text{and}\qquad
\mu_{\rm up}(Z;K)=0\,.
\]
One concrete construction is the following.
Place $x_{2^m}$ into $Z$ for every $m\geq 1$.
Let $r_1<r_2<r_3<\cdots$
be the positive integers that are not powers of two, and place $x_{r_q}$ into $A_i$ whenever $q\equiv i \pmod N$ (with residues taken in $[N]$).
Then each $A_i$ is infinite and the sets are pairwise disjoint.
Moreover,
\[
|Z\cap K^m|\leq 1+\floor{\log_2 m}=o(m)\,,
\]
so $\mu_{\rm up}(Z;K)=0$, and for every $i\in[N]$,
\[
\biggl|\,|A_i\cap K^m| - \frac{m-|Z\cap K^m|}{N}\,\biggr|\leq 1\,,
\]
which implies $\mu_{\rm up}(A_i;K)=1/N$.
For each $j\in[N]$, define
\[
L_j
\coloneqq
Z\cup \bigcup_{i:\, j\in S_i} A_i\,,
\qquadand
\cL\coloneqq \inbrace{K,L_1,\dots,L_n}\,.
\]
{As in the memoryless construction,} the $L_j$'s are distinct and infinite. 
Thus $|\cL|=n+1=k$.
Now fix, once and for all, enumerations: for each $1\leq i\leq N$
\[
A_i=\sinbrace{a_i^{(1)},a_i^{(2)},\dots}\,,
\]
and partition $Z$ into pairwise disjoint finite sets (for $r\geq 1$ and $1\leq i\leq N$) define $Z_{r,i}\subseteq Z$ such that 
\[
\forall_{r\geq 1 \text{~and~} 1\leq i\leq N}\,,\quad 
|Z_{r,i}|=r
\qquadand 
Z=\bigcup_{r\geq 1}\bigcup_{i=1}^N Z_{r,i}\,.
\]
Let $E_K$ be the repetition-free enumeration of $K$ obtained by listing these blocks stage by stage:
\[
a_1^{(1)}, Z_{1,1}, a_2^{(1)}, Z_{1,2}, \dots, a_N^{(1)}, Z_{1,N}
\quadand
a_1^{(2)}, Z_{2,1}, a_2^{(2)}, Z_{2,2}, \dots, a_N^{(2)}, Z_{2,N}\,,
\]
and so on, with each finite block written in an arbitrary order.
This is the fixed enumeration promised by the lemma.

Now let $W\geq 1$, and let $\generator$ be any window-$W$ set-based generator that generates from every language in $\cL$ in the limit under repetition-free enumerations.
For each language $L\in\cL$, apply \cref{lem:finite-bad-window-budget} to obtain a finite set $B_L\subseteq L$ such that every ordered $W$-tuple of distinct elements from $L\setminus B_L$ are contained in $L$.
Let
\[
B\coloneqq \bigcup_{L\in\cL} B_L\,.
\]
Since $\cL$ is finite, the set $B$ is finite.

Choose $r_0\geq W$ so large that every element of $B\cap K$ appears before stage $r_0$ of the fixed enumeration $E_K$.
Let $t^\ast$ be large enough that for every $t\geq t^\ast$, the current window $(x_{t-W+1},\dots,x_t)$ lies entirely inside stages $r\geq r_0$ of $E_K$.
Then every point in the current window belongs to $K\setminus B$.
Moreover, every separator block $Z_{r,i}$ occurring in stages $r\geq r_0$ has length at least $W$.
Consequently, a window of length $W$ can intersect at most one of the positive-density pieces $A_1,\dots,A_N$.
Hence for every $t\geq t^\ast$, exactly one of the following two cases holds.

\paragraph{Case 1 (The current window contains only elements of $Z$).}
Since $Z\subseteq L_j$ for every $j\in[N]$, the entire window is contained in each of the languages $K,L_1,\dots,L_n$.
Every element of the window also lies outside the corresponding bad set for each of these languages.
Therefore, by the definition of $B$ and \cref{lem:finite-bad-window-budget},
\[
G_t
\subseteq
K\cap \bigcap_{j=1}^n L_j\,.
\]
By construction, no set $S_i$ equals $[N]$, so no point of any $A_i$ belongs to all $L_j$ simultaneously.
On the other hand, every point of $Z$ belongs to all $L_j$.
Thus
\[
K\cap \bigcap_{j=1}^n L_j = Z\,,
\qquadtext{and, hence,}
\mu_{\rm up}(G_t;K)\leq \mu_{\rm up}(Z;K)=0\,.
\]

\paragraph{Case 2 (The current window contains points from exactly one set $A_i$, and possibly also points from $Z$).}
Every point of $A_i$ belongs exactly to the side languages $L_j$ with $j\in S_i$, while every point of $Z$ belongs to all side languages.
Hence the side languages that contain the entire window are precisely the languages $L_j$ with $j\in S_i$.
As above, the window lies outside the relevant bad sets, so
\[
G_t
\subseteq
K\cap \bigcap_{j\in S_i} L_j\,.
\]
We now compute this intersection.
Certainly $A_i\cup Z \subseteq K\cap \bigcap_{j\in S_i} L_j$.
Conversely, let $x\in K\cap \bigcap_{j\in S_i} L_j$.
If $x\in Z$, there is nothing to prove.
If $x\in A_r$ for some $r\in[N]$, then $x\in L_j$ holds exactly when $j\in S_r$.
Since $x$ lies in every $L_j$ with $j\in S_i$, we obtain $S_i\subseteq S_r$.
But both $S_i$ and $S_r$ have cardinality $\floor{n/2}$, so $S_i=S_r$, hence $r=i$.
Therefore
\[
K\cap \bigcap_{j\in S_i} L_j = A_i\cup Z\,.
\]
It follows that
\[
\mu_{\rm up}(G_t;K)
\leq
\mu_{\rm up}(A_i\cup Z;K)
\leq
\mu_{\rm up}(A_i;K)+\mu_{\rm up}(Z;K)
=
\frac{1}{N}\,.
\]
In both cases, for every $t\geq t^\ast$ we have
\[
\mu_{\rm up}(G_t;K)\leq \frac{1}{N}
=
\frac{1}{\binom{k-1}{\floor{(k-1)/2}}}\,.
\qedhere{}
\]
\end{proof}

\begin{remark}
{Note that the order of  quantifiers in \cref{lem:window-upper-density-collapse} is stronger than the minimax theorem needs.}
To show the minimax density bound, it suffices to have a result of the form 
\[
\forall~~\text{sliding-window sizes $W$}~~\; \exists~~\text{collection $\cL,$ and $K \in \cL$}~~\; \forall~~\text{generators $\generator$}~~ \;\exists ~~\text{``hard'' enumeration $E_K$}~~\,,
\]
so the hard collection and the {adversarial} repetition-free enumeration are allowed to depend on the window length.
By contrast, \cref{lem:window-upper-density-collapse} proves the  stronger statement
\[
\exists~~\text{collection $\cL, K \in \cL$ and ``hard enumeration $E_K$}~~\; \forall~~\text{sliding-window sizes $W$}\; \forall~~\text{generators $\generator$}~~\;\,.
\]
Thus the hard collection, target, and the adversarial enumeration do not depend on $W$.
A single fixed instance simultaneously {works against every finite sliding-window generator}.
In particular, the obstruction is not merely that each memory size $W$ has its own tailored counterexample; rather, there is one robust finite collection on which no finite sliding window improves on the memoryless Sperner bound.
\end{remark}
\noindent {The previous result gives the upper bound. The matching lower bound is immediate from the memoryless result by ignoring all but the most recent example. We now have all the ingredients we need to prove \cref{thm:window-upper-density-exact}.}

\begin{proof}[Proof of \cref{thm:window-upper-density-exact}]
Let $N\coloneqq \binom{k-1}{\floor{(k-1)/2}}.$
{If $k=1$, the unique target language $K$ can be output on every round, so $\rho_{\rm up}(1,W)=1=1/N$. Hence assume $k\geq 2$.}
The lower bound
$\rho_{\rm up}(k,W)\geq \sfrac{1}{N}$
follows from \cref{thm:sperner-achievability}: for any collection $\cL$ of size $k$, let $\generator_{\rm mem}:\cX\to[\cX]^\infty$ be the memoryless generator guaranteed by that theorem, and define the window-$W$ generator
$\generator_{\rm win}(y_1,\dots,y_W)\coloneqq \generator_{\rm mem}(y_W).$
Thus $\generator_{\rm win}$ simply ignores the first $W-1$ entries of the window and applies the memoryless rule to the most recent example.
On every repetition-free enumeration, the outputs of $\generator_{\rm win}$ from time $W$ onward are exactly the outputs of $\generator_{\rm mem}$ on the same current examples, so the same infinite set of good times witnesses upper density at least $1/N$.

The reverse inequality
$\rho_{\rm up}(k,W)\leq \sfrac{1}{N}$
is exactly \cref{lem:window-upper-density-collapse}.
Combining the two bounds proves the claim.
\end{proof}

\begin{remark}
 {The proof of \cref{lem:window-upper-density-collapse} shows why the sliding window fails to help: the adversary inserts long blocks of uninformative examples between the informative ones, so that every window contains at most one informative example. This, then allows us to apply the Sperner bound from the memoryless case.}
\end{remark}

\subsection{{Density Bounds for Generators with Adaptive Buffers}}
\label{sec:no-eviction-buffers}
{Finally, we consider the adaptive-buffer model, where the generator can store up to $b$ past examples of its choice.}
Throughout this section, we write $b$ for the buffer size and reserve $k$ for the size of the finite collection. We also restrict to repetition-free enumerations.

\begin{definition}[$b$-buffer set-based generator]
\label{def:no-eviction-buffer-generator}
Fix an integer $b\geq 0$, and let
$
\cM_b \coloneqq \bigcup_{r=0}^{b}\cX^r.
$
A \emph{$b$-buffer set-based generator} consists of two deterministic functions
\begin{align*}
    \mathrm{out}: \cM_b \times \cX \to [\cX]^\infty
    \quadand
    \mathrm{upd}: \cM_b \times \cX \to \cM_b\,,
\end{align*}
{such that, for every $M=(u_1,\dots,u_r)\in \cM_b$ and every $x\in \cX$, each entry of $\mathrm{upd}(M,x)$ lies in $\{u_1,\dots,u_r,x\}$.
Given an enumeration $E_K=(x_1,x_2,\dots)$ of a target language $K$, the buffer state $M_t\in \cM_b$ evolves from the empty tuple $()$ by}
\[
M_0\coloneqq ()\,,
\qquad
G_t \coloneqq \mathrm{out}(M_{t-1},x_t)\,,
\qquad
M_t \coloneqq \mathrm{upd}(M_{t-1},x_t)\,.
\]
\end{definition}
{In other words, the generator may keep, discard, reorder, or replace stored examples, but cannot synthesize new buffer contents that did not appear in the interaction. The buffer is the generator's only persistent state.}

{Similar as before, $\rho_{\rm up}(k,b)$ denotes the minimax upper density achievable by set-based generators with adaptive buffers of size $b$ against collections of size $k$.}
The main result in this section proves a lower bound on this quantity.

\begin{theorem}[{Lower bound for adaptive buffers}]
\label{thm:no-eviction-universal-lower-bound}
Let $k\in \N$ denote the size of a collection of languages and {$b\geq 0$ be} the size of a memory buffer.
Then, 
\[
\rho_{\rm up}(k,b)
\geq
\begin{cases}
\displaystyle \frac{1}{\binom{k-b-1}{\floor{(k-b-1)/2}}}\,, & 0\leq b\leq k-3\,,\\[3mm]
1\,, & b\geq k-2\,.
\end{cases}
\]
\end{theorem}

\begin{proof}
{Fix an arbitrary finite collection $\cL$ of size $k$.} We describe a $b$-buffer generator $\generator^{\rm buf}$.
For an example $x\in \cX$, let
\[
S(x)\coloneqq \inbrace{L\in \cL : x\in L}
\]
{denote the set of languages in $\cL$ that contain $x$, which we sometimes refer to as the ``global signature'' of $x$.}
Given a buffer
$M=(u_1,\dots,u_s),$
define its induced version space by
\[
\cL(M)\coloneqq \inbrace{L\in \cL : u_j\in L \text{ for every }j\in \insquare{s}}\,.
\]
Thus $\cL(())=\cL$.
{We sometimes refer to the induced version space as the residual collection.}

The update rule is greedy: when the current example is $x_t$, if the buffer is not yet full and
\[
\cL(M_{t-1})\cap S(x_t) \subsetneq \cL(M_{t-1})\,,
\]
then append $x_t$ to the buffer; otherwise leave the buffer unchanged.
Formally,
\[
\mathrm{upd}(M_{t-1},x_t)\coloneqq
\begin{cases}
(M_{t-1},x_t)\,, & \text{if }\abs{M_{t-1}}<b \text{ and } \cL(M_{t-1})\cap S(x_t)\subsetneq \cL(M_{t-1})\,,\\
M_{t-1}\,, & \text{otherwise.}
\end{cases}
\]
The output rule is the canonical intersection generator on the current residual collection:
if
$\generator_{\cap}^{\cA}(x)$
denotes the canonical memoryless intersection generator from the proof of \cref{thm:sperner-achievability}, applied to a finite collection $\cA$, then at time $t$ we output
$G_t \coloneqq \generator_{\cap}^{\cL(M_{t-1})}(x_t).$
This defines a valid $b$-buffer generator.

Fix a target language $K\in \cL$ and a repetition-free enumeration $E_K=(x_1,x_2,\dots)$ of $K$.
Every time the buffer changes, the residual collection $\cL(M_t)$ shrinks strictly.
Because $\cL$ is finite and the buffer has size at most $b$, the buffer stabilizes after finitely many rounds to some final state
\[
M_\ast=(u_1,\dots,u_s)\,,
\qquad s\leq b\,.
\]
We distinguish two cases.

\smallskip
\noindent
\textbf{Case 1 (The final buffer is not full):}
Assume $s<b$.
We claim that every language in $\cL(M_\ast)$ contains $K$.
Suppose not; then there exists some $L\in \cL(M_\ast)$ with $K\nsubseteq L$.
Choose a witness
$y\in K\setminus L,$
and let $t_y$ be the unique stage at which $y$ appears in the repetition-free enumeration of $K$.
Since the residual collections only shrink and $L$ survives in the final residual collection, we have
$L\in \cL(M_{t_y-1}).$
Because $y\notin L$, it follows that
$\cL(M_{t_y-1})\cap S(y)\subsetneq \cL(M_{t_y-1}).$
Also, since the construction never removes stored examples and the final buffer has size $s<b$, the buffer is not full at stage $t_y$.
Hence the update rule would append $y$ at time $t_y$, which would remove $L$ from all subsequent residual collections.
This contradicts $L\in \cL(M_\ast)$.
Therefore every language in $\cL(M_\ast)$ contains $K$.

Since $K\in \cL(M_\ast)$ as well, we have
$\bigcap_{L\in \cL(M_\ast)}L = K.$
Moreover, for every future current point $x_t\in K$, every language in $\cL(M_\ast)$ contains $x_t$, so the canonical intersection generator on $\cL(M_\ast)$ outputs exactly $K$.
Thus from the stabilization time onward the generator outputs $K$ on every round, and in particular it achieves upper density $1$.

\smallskip
\noindent
\textbf{Case 2 (The final buffer is full):}
Assume now that $s=b$.
Every insertion removed at least one language from the residual collection, so
$\abs{\cL(M_\ast)}\leq \abs{\cL}-b.$
From the stabilization time onward, the generator is exactly the canonical memoryless intersection generator on the fixed residual collection $\cL(M_\ast)$.

If $\abs{\cL(M_\ast)}\geq 3$, then by \cref{thm:sperner-achievability} this generator achieves upper density at least
\[
\frac{1}{\binom{\abs{\cL(M_\ast)}-1}{\floor{(\abs{\cL(M_\ast)}-1)/2}}}
\]
infinitely often on the target $K$, and {because $\binom{m-1}{\floor{(m-1)/2}}$ is nondecreasing in $m$ and $\abs{\cL(M_\ast)}\leq \abs{\cL}-b$,} this is at least
\[
\frac{1}{\binom{\abs{\cL}-b-1}{\floor{(\abs{\cL}-b-1)/2}}}\,.
\]
If $\abs{\cL(M_\ast)}\leq 2$, then the canonical memoryless intersection generator achieves upper density $1$ infinitely often: if the residual collection has size $1$, it outputs $K$ exactly from then on; if the residual collection has size $2$, say $\{K,L\}$, then on each current point $x\in K$ it outputs either $K$ or $K\cap L$, and either $K\setminus L$ is infinite (so $K$ is output infinitely often) or $K\cap L$ is cofinite in $K$ (hence has upper density $1$ in $K$).
So in this case the achieved upper density is again at least the claimed bound. 
\end{proof}

\begin{remark}[The proof does not use evictions]
\label{rem:no-eviction-lower-bound-no-evictions}
{Although \cref{def:no-eviction-buffer-generator} allows arbitrary evictions and replacements, the  algorithm in the proof of \cref{thm:no-eviction-universal-lower-bound} never uses them: it only appends a new example when doing so strictly shrinks the current residual collection, and once the buffer is full it keeps the stored tuple fixed forever. The lower bound therefore holds even for buffers that cannot evict stored examples.}
\end{remark}

\begin{remark}[Comparison with the memoryless bound]
\label{rem:no-eviction-memoryless-comparison}
When $b=0$, \cref{def:no-eviction-buffer-generator} reduces to the memoryless model, and
\[
\rho_{\rm up}(k,0) 
=
\frac{1}{\binom{k-1}{\floor{(k-1)/2}}}
\]
by \cref{thm:upper-density-collapse,thm:sperner-achievability}.
For general $b$, the lower bound from \cref{thm:no-eviction-universal-lower-bound} is exactly the memoryless Sperner bound with the collection size reduced from $k$ to $k-b$.
Equivalently, for $0\leq b\leq k-3$,
\[
\rho_{\rm up}(k,b)
\geq
\frac{\binom{k-1}{\floor{(k-1)/2}}}
{\binom{k-b-1}{\floor{(k-b-1)/2}}}
\cdot
\rho_{\rm up}(k,0)\,.
\]
{Thus, each stored example improves the universal guarantee by reducing the effective collection size by one.}
\end{remark}

\section{Identification with Last Guess: Incremental Learning}
\label{sec:incremental-bounded-memory} 
{In this section we study identification under last-guess memory.
Here, the learner forgets all past examples and retains only its previous output.
This is the classical incremental learning from positive data model introduced by \citet{lange1996incremental}.
It isolates the second memory resource from the introduction: past outputs rather than past examples.
The previous sections studied what generation and density guarantees remain when a generator can retain only limited information about past examples.
Here we ask what generation and identification guarantees remain when the learner retains only its most recent guess.}

\begin{definition}[Incremental learner]
\label{def:incremental-learner}
Fix an output space $\Omega$ and an initial output $\omega_0\in\Omega$. An incremental learner with output space $\Omega$ is a deterministic update function {$\learner:\Omega\times \cX\to\Omega$}.
Given an enumeration $E_K=\inparen{x_t\colon t\in\N}$, its outputs are defined recursively by
\[
\omega_t=\learner\!\inparen{\omega_{t-1},x_t}
\qquad\text{for every } t\geq 1\,.
\]
\end{definition}
{As in the earlier sections, the choice of output space determines the kind of model we obtain. We focus here on the index-based identification setting: $\Omega=\inbrace{1,\dots,N}$ for a finite collection $\cL=\inbrace{L_1,\dots,L_N}$, or $\Omega=\N$ for a countable collection, with output $i_t$ interpreted as the hypothesis language $L_{i_t}$. For finite collections, this means the learner has no extra internal states, synonym indices, or hypotheses outside the collection; its only persistent state is the previous output in $\inbrace{1,\dots,N}$.}

{Conceptually, the incremental model is more restricted than the (full-information) identification model of \citet{gold1967language}, which inspects the entire sample history. At the same time, it is stronger than a memoryless learner: an incremental learner can react differently to the same current example depending on its previous output.}
{Given this strength, one might ask if identification becomes possible in this model.}

Our first observation is that exact identification remains fragile in the incremental model, even for very small collections.

\begin{proposition}[No exact identification even for three languages]
\label{prop:incremental-not-identifiable}
There exists a collection of three infinite languages that is not identifiable in the limit by any incremental index-based learner {under arbitrary positive enumerations; indeed, the impossibility already holds for finitely repeating enumerations}.
\end{proposition}
The proof appears in \cref{sec:proofof:prop:incremental-not-identifiable}. 

Given this impossibility result, we turn to a natural relaxation of language identification in the limit: instead of requiring the learner to converge to the exact identity of the target language, we ask only that it converge to a language that differs from the target on at most finitely many elements. Since every language in the collection is infinite, two languages that differ on finitely many elements agree on all but a vanishing fraction of their elements. 
The learner has therefore captured essentially all of the target language, even if it has not pinpointed it exactly. 
{In particular, approximate identification guarantees that all notions of density studied in the previous section (\cref{sec:density-bounds}) take the optimal value $1$, and, hence, it is significantly stronger than lower bounding the generator's density.}

We formalize approximate identification through the following notion of almost-containment.
\begin{definition}[Almost-containment]
\label{def:almost-containment}
For $A,B\subseteq \cX$, write $A\preceq_{\rm F} B$ if $A\setminus B$ is finite, and write $A\sim_{\rm F} B$ if both $A\preceq_{\rm F} B$   and $B\preceq_{\rm F} A$ hold. Equivalently, $A\sim_{\rm F} B$ if $A\triangle B$ is finite.
\end{definition}
\begin{definition}[Approximate identification in the limit]
\label{def:approximate-identification}
Let $\cL=\inbrace{L_1,\dots,L_N}$, let $K\in\cL$, and let $\inparen{i_t}_{t\in\N}$ be the outputs of an incremental index-based learner on an enumeration of $K$. We say that the learner approximately identifies $K$ in the limit if there exists $t^\ast$ such that for every $t\geq t^\ast$,
\[
L_{i_t}\sim_{\rm F} K\,.
\]
We say that $\cL$ is approximately identifiable in the limit by incremental learners if some incremental index-based learner approximately identifies every $K\in\cL$ on every enumeration of $K$.
\end{definition}
{As a concrete analogy, imagine learning a mature programming language such as C from examples. Eventually one knows how to write ordinary C programs and all certain obsolete or exceptional corner cases. Such a learner has not identified the language exactly, but has learned almost all of it.}
This is precisely the kind of the learners the above definition of approximate identification intends to capture. 

Our main result for the incremental setting shows that this relaxed goal is always achievable for finite collections.

\begin{theorem}[Approximate identification for finite collections]
\label{thm:incremental-approximate-identification}
Every finite collection of infinite languages is approximately identifiable in the limit by an incremental index-based learner.
\end{theorem}
{\cref{thm:incremental-approximate-identification} stands in sharp contrast to \cref{prop:incremental-not-identifiable}: while \cref{prop:incremental-not-identifiable} shows that exact identification can fail already for a collection of size three, \cref{thm:incremental-approximate-identification} shows that approximate identification succeeds for every finite collection.}

{This result also gives one formal version of the contrast between the power of remembering past examples and past outputs. 
As shown in \cref{sec:density-bounds}, with memory of past examples, the best minimax density guarantee for finite collections deteriorates with the collection size no matter the number of examples the generator remembers.
In contrast, with memory of even the last guess, already suffices for enabling approximate identification of every finite collection and, hence, also for aching the optimal density regardless of the finite collection's size.}

It is also natural to ask which collections beyond the finite case admit approximate identification by incremental learners. The same argument extends to many countable families, provided the strict almost-containment relation admits a suitable topological ordering in which every target language has only finitely many predecessors. At the opposite extreme, if a collection violates the weak Angluin condition introduced by \cite{charikar2025facets,kalavasis2026characterizations}, then it is not approximately identifiable even by a full-information learner, and hence not by an incremental one either.

\subsection{Proof of \cref{prop:incremental-not-identifiable}}
\label{sec:proofof:prop:incremental-not-identifiable}

\begin{proof}
Let $C\coloneqq \inbrace{3n : n\in\N}$ and consider the collection
\[
L_1\coloneqq C\cup\inbrace{1}\,,\qquad
L_2\coloneqq C\cup\inbrace{2}\,,\qquad
L_3\coloneqq C\cup\inbrace{1,2}\,.
\]
Suppose toward a contradiction that some incremental index-based learner identifies this collection in the limit. Since the collection has exactly three languages, the learner has only three possible outputs, namely the indices $1,2,3${, and no additional states}.

For a finite text $\sigma$, let $q\inparen{\sigma}$ denote the learner's output after reading $\sigma$. Set
\[
\sigma_0\coloneqq \epsilon\,,\qquad
\sigma_1\coloneqq \inparen{1}\,,\qquad
\sigma_2\coloneqq \inparen{2}\,,\qquad
\sigma_{12}\coloneqq \inparen{1,2}\,,
\]
and let $T_C$ be any repetition-free enumeration of $C$. {The suffixes used below repeat one of the symbols $1,2$ at most once, so all resulting hard texts are still finitely repeating.} We claim that the four outputs
\[
q\inparen{\sigma_0}\,,\qquad q\inparen{\sigma_1}\,,\qquad q\inparen{\sigma_2}\,,\qquad q\inparen{\sigma_{12}}
\]
must be pairwise distinct.

{We use the following observation: if two finite prefixes leave the learner in the same output state, then appending the same suffix produces identical future outputs, because the previous output is the learner's only persistent state.}
Indeed, if $q\inparen{\sigma_1}=q\inparen{\sigma_2}$, then appending the same suffix $T_C$ yields identical future behavior on texts for $L_1$ and $L_2$, impossible. The same suffix $T_C$ shows that $q\inparen{\sigma_1}\neq q\inparen{\sigma_{12}}$ and $q\inparen{\sigma_2}\neq q\inparen{\sigma_{12}}$, since the corresponding full texts are texts for $L_1$ and $L_3$, and for $L_2$ and $L_3$, respectively. If $q\inparen{\sigma_0}=q\inparen{\sigma_1}$, then appending the common suffix consisting of $2$ followed by $T_C$ yields identical future behavior on texts for $L_2$ and $L_3$, again impossible. Similarly, $q\inparen{\sigma_0}\neq q\inparen{\sigma_2}$ by appending $1$ followed by $T_C$, and $q\inparen{\sigma_0}\neq q\inparen{\sigma_{12}}$ by appending $1$ followed by $T_C$.

Thus the learner would need at least four distinct outputs after these four finite texts, but only three outputs are available. This contradiction proves the proposition.
\end{proof}

\subsection{Proof of \cref{thm:incremental-approximate-identification}}
\label{sec:proofof:thm:incremental-approximate-identification}
 
\begin{proof}
Fix a finite collection $\cL=\inbrace{L_1,\dots,L_N}$, and define the strict almost-containment relation by
\[
A\prec_{\rm F} B
\qquad\text{if and only if}\qquad
A\preceq_{\rm F} B \text{ and } B\not\preceq_{\rm F} A\,.
\]
{The relation $\prec_{\rm F}$ is a strict partial order. Indeed, transitivity follows because if $A\prec_{\rm F}B$ and $B\prec_{\rm F}C$, then $A\setminus C\subseteq (A\setminus B)\cup(B\setminus C)$ is finite, while $C\setminus A$ is infinite since all but finitely many elements of the infinite set $C\setminus B$ lie outside $A$. Since $\cL$ is finite, choose a topological ordering and relabel so that}
\[
L_i\prec_{\rm F} L_j \quad\Longrightarrow\quad i<j\,.
\]
We consider the incremental learner with initial index $i_0\coloneqq 1$ and update rule
\[
i_t\coloneqq
\begin{cases}
i_{t-1}\,, & \text{if } x_t\in L_{i_{t-1}}\,,\\[1mm]
\min\inbrace{i_{t-1}+1,N}\,, & \text{if } x_t\notin L_{i_{t-1}}\,.
\end{cases}
\]
We will show that this learner approximately identifies every target language in $\cL$.

So fix a target language $K\in\cL$, let $z$ be the least index such that $L_z=K$, and let $E_K=\inparen{x_t}_{t\in\N}$ be any enumeration of $K$. By construction, the sequence $\inparen{i_t}_{t\in\N}$ is nondecreasing and takes values in the finite set $\inbrace{1,\dots,N}$. Hence it stabilizes: there exist $t_0$ and $i_\infty$ such that
\[
i_t=i_\infty
\qquad\text{for every } t\geq t_0\,.
\]

We first record a basic observation.

\begin{observation}\label{obs:incremental-hit-target}
If $i_t=z$ for some {$t\geq 0$}, then $i_s=z$ for every $s\geq t$.
\end{observation}
\begin{proof}
Once the learner outputs the index $z$, its current hypothesis is exactly $L_z=K$. Every later example belongs to $K$, and hence belongs to $L_z$, so the update rule never increments the index again. Thus $i_s=z$ for all $s\geq t$.
\end{proof}
We now prove the two almost-containment relations separately.

\begin{claim}\label{clm:incremental-first-direction}
It holds that $K\preceq_{\rm F} L_{i_\infty}$.
\end{claim}
\begin{proof}
We divide the proof into two cases.

\paragraph{Case 1: $i_\infty=N$.}
We claim that then necessarily $z=N$. Indeed, if $z<N$, then since the index sequence starts at $i_0=1$ and increases only by steps of size one, in order to reach $N$ it must pass through the value $z$ at some earlier time. But by \cref{obs:incremental-hit-target}, once the sequence reaches $z$, it can never move again. This contradicts $i_\infty=N>z$. Hence $z=N$, and therefore
\[
L_{i_\infty}=L_N=K\,,
\]
so in this case $K\preceq_{\rm F} L_{i_\infty}$ holds trivially.

\paragraph{Case 2: $i_\infty<N$.}
Take any element $y\in K\setminus L_{i_\infty}$. Let $t$ be the first time at which the enumeration presents $y$, so $x_t=y$. If $t>t_0$, then $i_{t-1}=i_\infty$ and $x_t\notin L_{i_\infty}$, so the update rule would force
\[
i_t=\min\inbrace{i_\infty+1,N}>i_\infty\,,
\]
contradicting the definition of $t_0$. Therefore every element of $K\setminus L_{i_\infty}$ must appear by time $t_0$. Since only finitely many examples appear by time $t_0$, the set $K\setminus L_{i_\infty}$ is finite. Hence $K\preceq_{\rm F} L_{i_\infty}$.
\end{proof}

\begin{claim}\label{clm:incremental-second-direction}
It holds that $L_{i_\infty}\preceq_{\rm F} K$.
\end{claim}
\begin{proof}
Suppose for contradiction that $L_{i_\infty}\not\preceq_{\rm F} K$. By \cref{clm:incremental-first-direction}, we already know that $K\preceq_{\rm F} L_{i_\infty}$. Thus the failure of the reverse containment implies that the containment is strict, namely
$K\prec_{\rm F} L_{i_\infty}.$
By the choice of the topological ordering of $\cL$, this forces
$z<i_\infty.$

But now the monotonicity of the index sequence gives a contradiction. Since $i_0=1\leq z$ and the sequence increases only by one at a time, in order to end at the value $i_\infty>z$, it must pass through the value $z$ at some intermediate stage. By \cref{obs:incremental-hit-target}, once the sequence reaches $z$, it can never move beyond $z$. This is impossible. Therefore our assumption was false, and we conclude that $L_{i_\infty}\preceq_{\rm F} K$.
\end{proof}
Combining \cref{clm:incremental-first-direction,clm:incremental-second-direction}, we obtain
$L_{i_\infty}\sim_{\rm F} K.$
Since $i_t=i_\infty$ for every $t\geq t_0$, it follows that
\[
L_{i_t}\sim_{\rm F} K
\qquad\text{for every } t\geq t_0\,.
\]
Thus the learner approximately identifies $K$ in the limit. Since $K\in\cL$ and its enumeration were arbitrary, the theorem follows.
\end{proof}

\section{Conclusion}

{In this paper, we study how bounded memory changes the hierarchy of generation, density, and identification in the limit. We show that the basic generation guarantee is surprisingly robust to memory restrictions. \cref{thm:memoryless-finite-reps} extends the guarantees of \citet{kleinberg2024language} to the memoryless setting: under the mild assumption of finitely-repeating enumerations, every countable collection admits a memoryless set-based generator. Moreover, requiring finite repetitions is necessary with bounded memory (\cref{thm:memoryless-with-reps}), though it is immaterial when the learner has unbounded memory.

Next, we turn to the stronger requirement of achieving density. Here the picture is more nuanced: memoryless generators admit a positive minimax upper density (\cref{thm:minimax-upper-density-memoryless-finite}) but no positive lower density once the collection has at least three languages (\cref{thm:no-uniform-positive-density}). 
We then ask whether giving the generator more memory of past examples improves the picture. The improvement is limited: a sliding window of any finite length gives no improvement at all (\cref{thm:window-upper-density-exact}), while an adaptive buffer does help by preserving informative examples once they appear (\cref{thm:no-eviction-universal-lower-bound}). 

The memory models we have considered so far all involve past examples --- keeping none of them, the most recent few, or a self-chosen subset. In the final part of the paper, we turn to a different resource: rather than past examples, the learner remembers only its previous output. Here the picture changes substantially. Although \cref{prop:incremental-not-identifiable} shows that exact identification can fail even for a collection of three languages, \cref{thm:incremental-approximate-identification} shows that a learner remembering only its most recent guess can approximately identify any finite collection. Since approximate identification gives density $1$ in the target under all the density notions considered here, remembering only the previous guess suffices for density $1$ on every finite collection  (\cref{thm:no-eviction-universal-lower-bound}) --- something no fixed amount of {sliding-window} memory can do. 
{Hence, the} type of memory used can matter more than the amount.

These results leave several directions open. One is to characterize the countable collections that are approximately identifiable with last-guess memory. Another is to understand what access to the collection is needed to implement bounded-memory generators: our constructions assume the algorithm has full access to the collection, and it remains unclear which oracle, computational, or representation assumptions preserve the same guarantees.}

\printbibliography 

\newpage

\appendix
\crefalias{section}{appendix}

\section{Further Results on Incremental Generation}
\label{sec:appendix:incremental-generation}

In this section, we continue our discussion of the incremental model.
We first show that exact identification and exact index-based generation can both fail for a collection of just three languages (\cref{sec:appendix:three-languages}).
We then show that this limitation is specific to index-based output: once the generator is allowed to output elements of the domain, the previous output can serve as a codeword that encodes the entire interaction history, effectively reducing the incremental model back to the full-information setting (\cref{sec:appendix:element-coding}).
Combining these ideas, we obtain a positive result for incremental element-based generation of every finite collection (\cref{sec:appendix:incremental-positive-generation}).

\subsection{Three Languages Already Prevent Exact Identification}
\label{sec:appendix:three-languages}
We begin by showing that exact identification in the incremental model can fail even for very small collections.

\begin{theorem}\label{thm:incremental-three-identification}
There exists a collection of three infinite languages that is not identifiable in the limit by any {incremental index-based learner}.
\end{theorem}
Note that any {incremental index-based learner} that eventually stabilizes to the correct index also succeeds as an incremental index-based generator for the same collection, so an impossibility for generation implies an impossibility for identification.
We prove the following stronger result about generation. 

\begin{proposition}\label{prop:incremental-three-generation}
Fix distinct elements $a,b,c\in \cX$ and an infinite set $T\subseteq \cX\setminus \inbrace{a,b,c}$. Define
\[
L_1\coloneqq T\cup \inbrace{a,b}\,,\qquad
L_2\coloneqq T\cup \inbrace{a,c}\,,\qquad
L_3\coloneqq T\cup \inbrace{b,c}\,.
\]
Then the collection $\cL\coloneqq \inbrace{L_1,L_2,L_3}$ is not generable in the limit by any incremental index-based generator. In fact, the impossibility already holds under finitely repeating enumerations.
\end{proposition}

\begin{proof}
Suppose toward a contradiction that $\generator:[3]\times \cX\to [3]$ is an incremental index-based generator for $\cL$, with initial state $i_0\in [3]$.
For a finite sequence $\sigma$, let $q(\sigma)$ denote the state reached after reading $\sigma$, and let $R$ be an injective enumeration of $T$.

The three languages are pairwise incomparable under inclusion. Therefore, if the target language is $L_r$ and the generator succeeds, then from some point onward it must output the index $r$ itself: no other language in the family is a subset of $L_r$. Hence whenever two finite prefixes can be followed by the same tail to produce valid enumerations for two different target languages, those two prefixes must lead to different states. 

Set $A\coloneqq q(a)$, $B\coloneqq q(b)$, and $C\coloneqq q(c)$. Comparing the pairs of texts $\inparen{a,c,R}$ and $\inparen{b,c,R}$, $\inparen{a,b,R}$ and $\inparen{c,b,R}$, and $\inparen{b,a,R}$ and $\inparen{c,a,R}$, we obtain that $A,B,C$ are pairwise distinct. Since there are exactly three states, this gives
\[
\inbrace{A,B,C}=[3]\,.
\]
Call a finite sequence completed for $L_1$ if every distinguished symbol it contains lies in $\inbrace{a,b}$ and both $a$ and $b$ appear; define completed sequences for $L_2$ and $L_3$ analogously. Appending the common tail $R$ to any completed sequence for $L_r$ yields a valid enumeration for $L_r$. Since completed prefixes for different target languages must land in different states and there are only three states altogether, there exist distinct states $p_1,p_2,p_3$ such that every completed prefix for $L_r$ lands in $p_r$. In particular,
\[
q(a,b)=q(b,a)=p_1\,,\qquad q(a,c)=q(c,a)=p_2\,,\qquad q(b,c)=q(c,b)=p_3\,,
\]
and also
\[
q(a,b,a)=q(a,b,b)=p_1\,,\qquad q(a,c,a)=q(a,c,c)=p_2\,,\qquad q(b,c,b)=q(c,b,c)=p_3\,.
\]
Next, $q(a,a)$ cannot equal $B$, since then the texts $\inparen{a,a,c,R}$ and $\inparen{b,c,R}$ would reach the same state before the common tail $\inparen{c,R}$, even though they are enumerations for $L_2$ and $L_3$. Similarly, $q(a,a)\neq C$, so $q(a,a)=A$. By symmetry, $q(b,b)=B$ and $q(c,c)=C$. Therefore
\[
\generator(A,a)=A\,,\qquad \generator(B,b)=B\,,\qquad \generator(C,c)=C\,,
\]
while the two-element completed prefixes give
\[
\generator(A,b)=p_1\,,\qquad \generator(A,c)=p_2\,,\qquad \generator(B,a)=p_1\,,
\]
\[
\generator(B,c)=p_3\,,\qquad \generator(C,a)=p_2\,,\qquad \generator(C,b)=p_3\,,
\]
and the three-element completed prefixes give
\[
\generator(p_1,a)=\generator(p_1,b)=p_1\,,\qquad
\generator(p_2,a)=\generator(p_2,c)=p_2\,,\qquad
\generator(p_3,b)=\generator(p_3,c)=p_3\,.
\]
Now $p_1\neq C$: otherwise the last display gives $\generator(C,a)=C=\generator(C,b)$, while the previous display gives $\generator(C,a)=p_2$ and $\generator(C,b)=p_3$, forcing $p_2=p_3=C$, a contradiction since $p_2$ and $p_3$ are distinct. By the same argument, $p_2\neq B$ and $p_3\neq A$. Since $\inbrace{p_1,p_2,p_3}=\inbrace{A,B,C}$, only two cases remain:
\[
\inparen{p_1,p_2,p_3}=\inparen{A,C,B}
\qquad\text{or}\qquad
\inparen{p_1,p_2,p_3}=\inparen{B,A,C}\,.
\]
In the first case, the transition table entries for $A,B,C$ on inputs $\inparen{a,b,c}$ are
\[
A:\inparen{A,A,C}\,,\qquad B:\inparen{A,B,B}\,,\qquad C:\inparen{C,B,C}\,,
\]
while in the second case they are
\[
A:\inparen{A,B,A}\,,\qquad B:\inparen{B,B,C}\,,\qquad C:\inparen{A,C,C}\,.
\]
On the other hand, by definition of $A,B,C$, the transition entry of the initial state $i_0$ must be
$\inparen{A,B,C},$
since $\generator(i_0,a)=A$, $\generator(i_0,b)=B$, and $\generator(i_0,c)=C$. Because $i_0\in [3]=\inbrace{A,B,C}$, its entry must match one of the three entries listed above, but in neither case does any entry equal $\inparen{A,B,C}$. This contradiction completes the proof.

Finally, every enumeration used in the proof repeats each distinguished symbol at most twice and enumerates $T$ injectively, so the same impossibility holds under finitely repeating enumerations.
\end{proof}
{This also proves \cref{thm:incremental-three-identification}: if an incremental index-based learner identified the displayed collection, then, because the three languages are pairwise incomparable, its eventual correct index would be a valid eventual index-based generator output for each target, contradicting \cref{prop:incremental-three-generation}.}

\subsection{Element-Based Generation Can Simulate Infinite Memory via Coding}
\label{sec:appendix:element-coding}
The impossibility in \cref{prop:incremental-three-generation} relies crucially on the fact that the output is an index from a finite set $[N]$. 
An index carries at most $\log{N}$ bits, so after each round the generator retains only a bounded amount of information about the past.
(Note that, as we showed in \cref{thm:incremental-approximate-identification}, if one relaxes the requirement of identification slightly to allow for a finite number of ``errors,'' then surprisingly the $\log{N}$ bits retained from the indices are sufficient.)

The situation changes completely once the generator is allowed to output an element of the domain $\cX$. The key observation is that the previous output is fed back to the generator at the next round, and an element of $\cX$ can encode an arbitrarily long string. The generator can therefore use its output as a codeword: it embeds a description of the entire observed prefix into the identity of the output element, and at the next round it decodes this description before choosing a new codeword. In this way, a single output element can carry enough information to simulate a full-information learner, and the bounded-memory restriction effectively collapses. 
From the perspective of bounded-memory learning, this is undesirable: the model no longer limits how much information is retained across rounds, but only hides that information inside the representation of the last output. The apparent ``memory bound'' is therefore largely an artifact of the output format rather than a genuine restriction on the learner, which is why we do not study this model further.

We make this precise using a construction of subsets below. For each language $L_i$ in the collection, we select an infinite subset $C_i\subseteq L_i$ such that the subsets are pairwise disjoint and each $C_i$ is cofinal\footnote{A subset $C\subseteq \cX$ is cofinal (with respect to an ordering $\prec$) if for every $x\in \cX$ there exists $y\in C$ such that $x\prec y$.} in the canonical ordering of $\cX$. The identity of the subset encodes the current hypothesis or hidden state, while the position of the output within that subset encodes any additional finite information, including the entire observed prefix. 
Thus the previous output serves as {an unbounded storage device} for the past (which is, as we mentioned, undesirable).

\begin{lemma}[Disjoint infinite cofinal subsets]
\label{lem:incremental-cofinal-subsets}
For every finite collection $\cL=\inbrace{L_1,\dots,L_N}$ of infinite languages, there exist sets $C_1,\dots,C_N\subseteq \cX$ such that, for every $i\in [N]$, the set $C_i$ is an infinite subset of $L_i$, the sets $C_1,\dots,C_N$ are pairwise disjoint, and for every $x\in \cX$ there exists $y\in C_i$ with $x\prec y$.
\end{lemma}

\begin{proof}
Write the canonical ordering of $\cX$ as $\inparen{\bar{x}_1,\bar{x}_2,\dots}$. For each $t\geq 1$ and $i\in [N]$, choose an element $c_{i,t}\in L_i$ recursively as follows. Let
\[
U_{t,i}\coloneqq \inbrace{c_{j,r}: 1\leq r\leq t-1,\ j\in [N]}\cup \inbrace{c_{j,t}: 1\leq j\leq i-1}\,.
\]
{Since only finitely many elements of $\cX$ lie at or before $\bar{x}_t$ in the canonical ordering, the infinite set $L_i$ contains infinitely many elements above $\bar{x}_t$; removing the finite set $U_{t,i}$ still leaves some $y\in L_i\setminus U_{t,i}$ with $\bar{x}_t\prec y$. Let $c_{i,t}$ be the first such $y$ in the canonical ordering. After making these choices for all $t$ and $i$, set $C_i\coloneqq \inbrace{c_{i,t}: t\in\N}$. Then each $C_i$ is an infinite subset of $L_i$, the sets $C_1,\dots,C_N$ are pairwise disjoint by construction, and for any $x\in \cX$, taking $t$ with $x=\bar{x}_t$ gives $x\prec c_{i,t}\in C_i$.}
\end{proof}
\noindent Write each $C_i$ in increasing order as
{$C_i=\inbrace{d_{i,1}\prec d_{i,2}\prec d_{i,3}\prec \cdots}.$}
Since the sets are pairwise disjoint, every element of $\bigcup_{i\in [N]} C_i$ uniquely determines both its language index $i$ and its position within $C_i$.
This is the mechanism that enables coding: the generator can read off both the current hypothesis and the encoded history from a single output element.
{In this section, an incremental element-based generator means an incremental learner with output space $\cX$: it has an initial memory state $s_0\in\cX$, updates by $s_t=\generator(s_{t-1},x_t)$, and succeeds on a target $K$ if, for all sufficiently large $t\geq 1$, $s_t\in K\setminus S_t$, where $S_t=\inbrace{x_1,\dots,x_t}$. The initial state $s_0$ is not itself required to be a valid generated element.}

\begin{theorem}[Coding compilation]
\label{thm:incremental-coding-compilation}
Let $\cL=\inbrace{L_1,\dots,L_N}$ be a finite collection of infinite languages. Suppose $M$ is an index-based learner in the unrestricted full-information model with the following property: for every target language $K\in \cL$ and every enumeration $\inparen{x_t}_{t\in\N}$ of $K$, if $i_t$ denotes the index output by $M$ after seeing $\inparen{x_1,\dots,x_t}$, then there exists $t_0$ such that
\[
L_{i_t}\preceq_{\rm F} K\qquad\text{for all } t\geq t_0\,.
\]
Then $\cL$ admits an incremental element-based generator.
\end{theorem}

\begin{proof}
{Fix an encoding $\operatorname{code}:\cX^{<\omega}\to \N$ of finite sequences by natural numbers, and a bijection $\langle\cdot,\cdot\rangle:\N\times \N\to \N$ such that, for every fixed $u\in\N$, the set $\inbrace{\langle u,n\rangle:n\in\N}$ is unbounded.} We define an incremental element-based generator whose previous output stores the entire observed prefix.

Initialize with
\[
s_0\coloneqq {d_{1,\langle \operatorname{code}(\epsilon),1\rangle}}\,.
\]
On round $t\geq 1$, suppose the previous output is $s_{t-1}$. {The update rule is defined arbitrarily on previous outputs outside $\bigcup_{i\in[N]}C_i$. Along the intended run, the previous output is always a codeword. If $s_{t-1}=d_{i,m}$, the disjointness of the $C_i$ and the displayed ordering of each $C_i$ determine a unique pair $(i,m)$; the update rule unpairs $m=\langle u,n\rangle$ and decodes $u$ to recover the previously seen prefix}
$\sigma_{t-1}=\inparen{x_1,\dots,x_{t-1}}.$
It then appends the new datum $x_t$ to form $\sigma_t=\inparen{x_1,\dots,x_t}$, runs the full-information learner on this prefix to obtain
$i_t\coloneqq M(\sigma_t),$
and lets $n_t$ be the least $n\in \N$ such that
\[
s_t\coloneqq {d_{i_t,\langle \operatorname{code}(\sigma_t),n\rangle}}
\]
lies after both $s_{t-1}$ and $x_t$ in the canonical ordering. {Such an $n_t$ exists because, for fixed $u=\operatorname{code}(\sigma_t)$, the positions $\inbrace{\langle u,n\rangle:n\in\N}$ are unbounded, so the corresponding subsequence of the increasing enumeration of $C_{i_t}$ is cofinal in $C_{i_t}$, and $C_{i_t}$ is cofinal in $\cX$.}

By construction, $s_t$ again determines $\sigma_t$, so the process can continue inductively. Also, $s_t\succ s_{t-1}$ and $s_t\succ x_t$ for every $t\geq 1$. Hence $\inparen{s_t}_{t\in\N}$ is strictly increasing in the canonical order, and an induction on $t$ shows that
\[
s_t\notin S_t\qquad\text{for every } t\geq 1\,,
\]
where $S_t=\inbrace{x_1,\dots,x_t}$.

Now fix a target language $K\in \cL$ and an enumeration $\inparen{x_t}_{t\in\N}$ of $K$. By hypothesis, there is $t_0$ such that $L_{i_t}\preceq_{\rm F} K$ for all $t\geq t_0$. Since $\cL$ is finite, the set
\[
B\coloneqq \bigcup \inbrace{L_i\setminus K : L_i\preceq_{\rm F} K}
\]
is finite. For every $t\geq t_0$, we have $s_t\in C_{i_t}\subseteq L_{i_t}\subseteq K\cup B$. Because the outputs are strictly increasing in the canonical order, they eventually lie beyond every element of $B$, and in particular there is $t^\ast\geq t_0$ such that $s_t\notin B$ for all $t\geq t^\ast$. For those $t$ we therefore have $s_t\in K$, and since also $s_t\notin S_t$, it follows that
\[
s_t\in K\setminus S_t\qquad\text{for all } t\geq t^\ast\,.
\]
Thus the incremental element-based generator succeeds on $K$. Since $K$ was arbitrary, the theorem follows.
\end{proof}
To summarize, \cref{thm:incremental-coding-compilation} shows that any full-information learner whose hypotheses are eventually almost contained in the target can be compiled into an incremental element-based generator. The compilation works because each output element encodes the full interaction history, allowing the incremental generator to simulate the full-information learner step by step.

\subsection{A Positive Result for Generation}
\label{sec:appendix:incremental-positive-generation}

We now combine the coding compilation of \cref{thm:incremental-coding-compilation} with the approximate-identification theorem from \cref{sec:incremental-bounded-memory} to obtain an incremental element-based generator for finite collections.

This result should be interpreted with caution: it is possible only because element-based generators can encode arbitrarily large amounts of information in their last action. 
Thus, while the result is positive, it does not meaningfully reflect the strength of bounded-memory algorithms, and for this reason we do not explore element-based generators which have access to their last guess further.

\begin{theorem}\label{thm:incremental-element-generation}
Every finite collection of infinite languages admits an incremental element-based generator.
\end{theorem}

\begin{proof}
By \cref{thm:incremental-approximate-identification}, every finite collection of infinite languages admits an incremental learner whose hypotheses eventually differ from the target language on only finitely many elements. {In particular, writing $i_t$ for the learner's output, eventually $L_{i_t}\sim_{\rm F}K$, and hence $L_{i_t}\preceq_{\rm F}K$, which is the one-sided hypothesis needed in \cref{thm:incremental-coding-compilation}.} Since any incremental learner is also a full-information learner, \cref{thm:incremental-coding-compilation} applies and yields an incremental element-based generator for the same collection.
\end{proof}

\end{document}